%% file: main.tex
\documentclass[letterpaper,11pt,twoside]{article}

\usepackage{authblk}

\usepackage{graphicx} 
\usepackage[margin=1in]{geometry}

\input{preamble.sty}
\bibliography{refs}

\title{Deterministic Incremental Triconnectivity in $O(\log^2 n)$ Worst-Case Time\footnote{Related version: This paper is based on the M.Sc. thesis by Niels Inselmann Christiansen, supervised by Eva Rotenberg, defended February 2026 at the Technical University of Denmark.}}

\author{
Niels Inselmann Christiansen \footnote{This research is supported by Danmarks Frie Forskningsfond via the grant ``Dynamic Graphs: Distributed and Geometry'' with grant-ID 10.46540/4251-00004B.}}
\author{Eva Rotenberg$^\dagger$
}

\affil{IT University of Copenhagen}

\date{} 

\begin{document}
\maketitle
\thispagestyle{empty}
\vspace{5em}
\begin{abstract}
We consider the problem of maintaining the triconnectivity in a graph $G$ in the incremental dynamic setting. We consider updates in the form of edge insertions and edge subdivisions, and triconnectivity queries, i.e. queries that take as input variables the vertices $a$ and $b$ and are of the form ``are there at least three vertex disjoint paths between $a$ and $b$?''. All algorithms presented in this paper are deterministic.

When starting from a biconnected graph, all updates take worst-case $O(\log n)$ time, and all queries are answered in $O(\log n)$ time. By combining this result for biconnected graphs with a result of Holm and Rotenberg [SODA'20] we get a deterministic algorithm for general graphs with a worst-case update and query time of $O(\log^2 n)$. Both algorithms represent an improvement of the state-of-the-art by a factor of $O(\log n)$.
%
\end{abstract}

\newpage

\setcounter{page}{1}
\pagestyle{fancy}
\input{tex/05introduction}
\input{tex/12heavypath}

\input{tex/20incremental}
\input{tex/30IncrementalWorstCase}

\input{tex/40ending}

\input{tex/input-til-referenceliste}
\printbibliography

\end{document}

%% file: tex/05introduction.tex

\section{Introduction \& preliminaries}

In a graph, two vertices are \emph{triconnected} (also denoted 3-vertex-connected) if there exist three internally vertex disjoint paths between them. This relation induces a covering on the vertices of the graph, the member sets of which we call the \textit{triconnected components} of G. 
The triconnected components of a graph have received much attention due to their importance in graph drawing: a triconnected graph has a unique planar embedding, and thus, the triconnected components give information about any planar embedding of a planar graph.
In the static setting, a linear time algorithm for finding the triconnected components is due to Hopcroft and Tarjan~\cite{Hopcroft1973}%
. 
In this work, we
consider the \textit{incremental} variant of this problem where edges and vertices arrive interspersed with queries about the triconnected components of the graph.
The analogous question about graph connectivity in the incremental setting is solved by the well-established union find data structure, which has tight $\Theta(\alpha(n))$ amortised and $\Theta(\frac{\log n}{\log\log n})$ worst-case time per operation. For triconnectivity, an $O(\log(n))$ amortised update time has been established for decades \cite{Battista1996}, but the first polylogarithmic incremental worst-case algorithm was only introduced in 2020~\cite{HolmRotenberg2019}. In this work, we take a major step towards obtaining tight bounds in the worst case setting also: 
Given a biconnected graph, we give $O(\log n)$ deterministic worst-case bounds for updating triconnectivity information under edge-insertions, which is tight up to $O(\log\log n)$-factors.
Our algorithm also leads to an improvement for general graphs, improving from worst-case $O(\log^3 n)$ to $O(\log^2 n)$.


\subsection{Related Work}
Vertex connectivity and the triconnectivity in graphs has received attention from the graph theory community for over a century~\cite{menger,maclane1937,Whitney1932,tutte66}. 
On the computational front, Hopcroft and Tarjan~\cite{Hopcroft1973} gave linear time algorithms for dividing a general graph into biconnected components and dividing a biconnected graph into triconnected components. 

An algorithm solving the incremental connectivity (commonly attributed to McIlroy and Morris) was proved to have an amortised running time of $O(\alpha (n))$ \cite{Tarjan:75}, which is optimal in the cell-probe model \cite{tarjan79,fredmansaks89}. This lower bound extends to bi- and triconnectivity with a simple reduction and has matching algorithms by Westbrook and Tarjan \cite{WestbrookT92} and Di Battista and Tamassia~\cite{Battista1996}. The latter~\cite{Battista1996} introduced \textit{SPQR-trees} as a data structure for storing the triconnected components of a graph and gave procedures for maintaining this under subdivisions and edge insertions. Similar to SPQR-trees, Tutte~\cite{tutte66} had previously introduced a decomposition (now dubbed the \emph{Tutte decomposition}) that reflects the relation between triconnected components and separation pairs in the graph; this Tutte-decomposition can be computed in linear time~\cite{Hopcroft1973,GutwengerM00,Planken25}. For incremental worst-case time complexity, there is a lower bound for incremental connectivity of $\Omega(\frac{\log n}{\log \log n})$ \cite{blum86,fredmansaks89} which is tight with matching algorithm due to Alstrup, Ben-Amram and Rauhe \cite{alstrupBR99}. In the realm of incremental worst-case update-time algorithms, the current state-of-the-art for biconnectivity is $O(\log n)$ due to La Poutré and Westbrook \cite{PoutreW98} and $O(\log^3 n)$ for triconnectivity due to Holm and Rotenberg \cite{HolmRotenberg2019}.

The fully dynamic connectivity problem was first studied by Frederickson \cite{Frederickson85} and has since received much attention by the community \cite{Eppstein97,HenzingerK99,HeTh97,HolmLT01,Thorup00,Eppstein:2003,patrascu06,Wulff-Nilsen16a,ChuzhoyGLNPS20,Kapron:2013}. It has a lower bound by Patrascu and Demaine \cite{patrascu06} in the cell-probe model of $\Omega(\log n)$ for the maximum amongst update time or query time, and current state of the art of amortised, randomised $O((\log n)(\log \log n)^2)$ due to Huang, Huang, Kopelowitz, Pettie and Thorup \cite{HuangHKP17}. The fully dynamic biconnectivity problem has also received attention \cite{WestbrookT92,Rauch92,Rauch94,Henzinger95,Henzinger95fully,Eppstein99,Henzinger00} together with the related 2-edge-connectivity problem \cite{WestbrookT92,Frederickson97,Eppstein99,HolmLT01,Henzinger97,BridgeFinding,HolmRotenberg2019}. The triconnectivity problem, supporting relational $\operatorname{triconnected}(u,v)$ queries, has yet to admit a polylogarithmic time algorithm, and has been studied for planar graphs~\cite{Eppstein97,DecrementalSPQR}. 

Many of these algorithms maintain a dynamic forest as a skeleton for the graph using e.g. Sleator and Tarjan's Link-Cut Trees \cite{SleatorTarjan1983} or Top-Trees by Alstrup, Holm, De Lichtenberg and Thorup \cite{top1}. These allow, among other things, $O(\log n)$-time insertion of edges and vertices and adding weight to a path of the forest. There are examples of these operations being supported using a weighted variant of search trees such as Bent, Sleator and Tarjan's biased search trees \cite{BentSleatorTarjan1985} or Sleator and Tarjan's splay trees \cite{SleatorT85}, which allow for a 'trees-within-trees' structure that maintains a total logarithmic depth for any leaf node.

\textbf{Subsequent and parallel work.} Recent results give efficient randomised, amortised algorithms for general vertex connectivity of a graph~\cite{Adena26}. Note here that our results are deterministic and that the update time guarantees hold in the worst-case%
, and that we answer pair-wise relational queries.

\subsection{Our contribution}
We study the maintenance of connected and biconnected graphs. Let $G$ be a biconnected simple graph, we give a data structure which can maintain a representation of the triconnected components of $G$ while supporting the following updates and query:
\begin{itemize}
    \item \V{Subdivide$(e)$}: For $e \in E(G)$. Replace the edge $e=uv$ with a path $uxv$ through a new vertex $x$. 
    \item \V{InsertEdge$(u,v)$}: For $(u,v) \notin E(G)$. Connect $u$ and $v$ with a new edge.
    \item \V{Triconnected$(u,v)$}: Are the vertices $u$ and $v$ triconnected?
\end{itemize}
It is well known every 2-connected graph can be built from the double bond ($K_2^2$) using only these operations.
In \cref{spqr2} we prove
\newcommand{\theoremonetext}{
There is a data structure that maintains a biconnected graph $G$ in linear space, while supporting the operations \V{Subdivide} and \V{InsertEdge} in deterministic worst-case $O(\log n)$ time and answers \V{Triconnected} in $O(\log n)$ deterministic worst-case time. 
}
\begin{bigtheorem}\label{result:biconn}
\theoremonetext
\end{bigtheorem}\noindent
Combining this with a theorem by Holm and Rotenberg \cite{HolmRotenberg2019} we consequently obtain 
\newcommand{\theoremtwotext}{
There is a data structure that maintains a graph on $n$ vertices in linear space, while supporting the operations \V{Subdivide} and \V{InsertEdge} in deterministic worst-case $O(\log^2 n)$ time and answers \V{Triconnected} in deterministic worst-case $O(\log^2 n)$ time. 
}
\begin{bigcorollary}\label{result:general}
\theoremtwotext
\end{bigcorollary}\noindent

\subsection{Techniques and challenges}
Our paper mainly builds on the work of Holm and Rotenberg \cite{HolmRotenberg2019}, which further build on the dynamic trees of Sleator and Tarjan \cite{SleatorTarjan1983} and the incremental SPQR-tree of Di Battista and Tamassia \cite{Battista1996}. Holm and Rotenberg reduce the problem of maintaining the SPQR-tree of some graph, to maintaining a tree under edge insertion, edge deletion and path contraction efficiently. For this purpose they use ST-trees along with an edge marking system, to indicate contracted edges. These edge markings are a classical application of ST-trees and can be implemented without additional asymptotic cost, however edge insertion in the SPQR-tree requires finding a specific edge, the search procedure of which requires an additional factor of $O(\log n)$ in the running time. We describe how one, by storing additional information in the ST-tree, can find these important edges without a slowdown. \cref{hp} is dedicated to recalling ST-trees and related techniques. \cref{spqr} defines SPQR-trees and the tree operations required to maintain them, and \cref{spqr2} restates the reduction of Holm and Rotenberg, ending with a discussion of the representation of the tree and how to implement all operations in $\BigO(\log n)$ time for biconnected graphs.

\subsection{Preliminaries}
Let $G=(V(G),E(G))$ be a (multi)graph with vertex set $V(G)$ and edge (multi)set $E(G)$. Throughout this paper we use the convention that $n=|V(G)|$ and $m=|E(G)|$. A tree $T$ is a connected graph with no cycles. It is well known that trees have exactly $m=n-1$ edges and that there is a unique path between any two vertices of $T$. In this paper, we will denote the vertices of a tree as nodes. A tree rooted at a node $r$ is a directed tree where all edges point in the direction of the unique path to $r$. All nodes $v$ (except the root) have a neighbour closest to the root; we call this the parent, $p(v)$. All neighbours which are not the parent are children of $v$. The set of children is denoted $c(v)$. The set of nodes which contain $v$ on their root path are the descendants of $v$. The subtree of $T$ rooted at $v$, denoted $T_v$, is $v$ and all its descendants. 

The graph on $n$ vertices where all pairs share an edge is called the complete graph on $n$ vertices, $K_n$. We call the multigraph on exactly $2$ vertices with $m$ distinct edges between them $K_2^m$.

\textbf{Connectivity}
A graph $G$ is $k$-connected if there are $k$ internally-vertex-disjoint paths between any two vertices of G. A 1-connected, 2-connected or 3-connected graph is generally called connected, biconnected or triconnected. Menger's theorem \cite{menger} states that for a $k$-connected graph, no set $c \subseteq V(G)$ of size $k-1$ will leave $G-c$ disconnected. Such a set is generally known as a \textit{(k-1)-cut} and serves as a certificate that a graph is not $k$-connected. A 1-cut is commonly called a cut-vertex and a 2-cut a separation pair. 

This motivates the decomposition of graphs into biconnected components.
\begin{definition}
Given a graph $G$, its biconnected components (or blocks) are:
\begin{enumerate}[label=(\roman*)]
    \item Its maximal biconnected subgraphs
    \item The bridges (cut-edges) of $G$ (trivial blocks)
\end{enumerate} 
\end{definition}
The idea to decompose a graph using $k$-separators has since been applied to define the triconnected components of a graph which we will revisit in \cref{spqr}.

%% file: tex/12heavypath.tex
\section{Heavy Path Decompositions }\label{hp}
In this section we will recall the dynamic trees of Sleator and Tarjan \cite{SleatorTarjan1983} in the weighted setting, which we will use later to store the triconnected components of a biconnected graph, while supporting changes to the tree in logarithmic worst-case time. We present the weighted formulation used by Holm and Rotenberg \cite{HolmRotenberg2019}.

\begin{definition} 
    Let $w$ be a positive weighting of the nodes of a rooted tree $T$. Define the weight $\w(v)$ of the subtree $T_v$ rooted at $v$ to be 
    $$
    w(T_v) = \w(v) :=  w(v) + \sum_{u \in c(v)} \w(u). 
    $$
    Then a (proper) heavy path decomposition of $T$ is a partitioning of $E(T)$ into two groups, \textit{heavy} and \textit{light} such that the edge  $(v,p(v))$ is heavy if and only if
    $$
    \w(v) > \frac{\w(p(v))}{2}
    $$
    and light otherwise.
\end{definition}

Apart from splitting the edges, this also partitions $T$ into so-called \textit{heavy paths} which we will use as a kind of compressed or generalised edge. This partitioning can be done in linear time using a simple traversal of $T$ by computing $\w$ bottom up. This definition, however, is somewhat limited as the partitioning of heavy and light is fixed for any tree. To mitigate this, we will consider a variant which enables us to force a single path to be heavy.

\begin{definition}
    A $u$-exposed heavy path decomposition of a rooted tree $T$ is a partitioning of $E(T)$ into two groups, \textit{solid} and \textit{dashed} such that the edge $(v,p(v))$ is solid if and only if it is on the $u$-to-root path or $p(v)$ is not on the $u$-to-root path and $(v,p(v))$ is heavy. 
\end{definition}

We can consider a solid edge to be almost heavy and a dashed edge to be almost light. We note that a proper heavy path decomposition is $u$-exposed for some choice of $u$. We now prove some properties of heavy path decompositions and describe some internal operations.

\begin{lemma}\label{hp:oneheavy}
    A vertex $v$ has at most one heavy child.
\end{lemma}
\begin{proof}
    Assume that $(c_1,v)$ and $(c_2,v)$ are heavy. Then $$
    \w(c_1) + \w(c_2) > \w(v).
    $$
    However, this is contradicted by
    $$
    \w(v)=w(v)+\sum_{u\in c(v)}\w(u)\geq \w(c_1)+\w(c_2).
    $$
\end{proof}

\begin{corollary}
    A node has at most one solid child.
\end{corollary}

\begin{lemma}\label{hp:lightdepth}
    Let $r$ be the root of $T$ and $W=\w(r)$. For any node $v$ in a heavy path decomposition the light depth is at most $\log(\frac{W}{\w(v)})$.
\end{lemma}
\begin{proof}
    Consider the light edges of the $v$-to-root path and let $l_1,\dots,l_m$ be the light children of said edges. Clearly
    $$
    2\w(l_i) < \w(p(l_i)) \leq \w(l_{i+1}).
    $$
    Applying this inequality along the entire path, we get
    $$
    2^m\w(v)\leq 2^m\w(l_1) < 2^{m-1}\w(l_2) < \dots < \w(l_m) \leq \w(r).
    $$
    Thus $$
    m \leq \log(\w(r)) - \log(\w(v)) =\log\left(\frac{W}{\w(v)}\right).
    $$
\end{proof}

\begin{lemma}\label{hp:dashedepth}
    The dashed depth of a node $v$ in a $u$-exposed heavy path decomposition of $T$ is at most 1 more than the light depth. 
\end{lemma}
\begin{proof}
    On the path from $v$ to the first node $x$ with a parent on the $u$-to-root path, the notion of light and dashed is equivalent. The path from $p(x)$ is entirely solid. Thus, only the edge $(x,p(x))$ might be heavy but not solid. This means that the dashed depth is at most 1 more than the light depth. 
\end{proof}\noindent
With this, we are ready to describe the internal updates. Suppose we have a heavy path decomposition $T$. We will support the following operations on $T$. Let $u$ be a vertex in $T$ and $r$ the root of $T$.
\begin{itemize}
    \item \V{expose}$(u)$: If $T$ is proper, make $T$ $u$-exposed.
    \item \V{conceal}$(u)$: If $T$ is $u$-exposed, make $T$ proper. 
    \item \V{reverse}$(r)$: If $T$ is $u$-exposed, make $u$ the root and $T$ $r$-exposed.
    \item \V{link-exposed}$(u,v)$: Given node $u$ in a $u$-exposed heavy path decomposition and a root node $v$ in another, make $u$ the dashed parent of $v$ resulting in a new $u$-exposed heavy path decomposition containing the nodes of both trees.
    \item \V{cut-exposed}$(u)$: Where $T$ is $u$-exposed and $u$ is not the root $r$ of $T$, remove the edge $(u,p(u))$ creating two new trees $T_r$ and $T_u$. $T_r$ is $p(u)$-exposed and $T_u$ is $u$-exposed.
\end{itemize}\noindent
To do this, we also need the following 2 internal operations used to edit the tree locally. Let $h(u)$ be the (unique) solid child of $u$ if it exists.
\begin{itemize}
    \item \V{splice}$(u)$: Where $u$ is a light child of its parent. If $(u,p(u))$ is dashed, make it solid. If $p(u)$ has a solid child, make it dashed. 
    \item \V{slice}$(u)$: Where $u$ is a light child of $p(u)$. If $(u,p(u))$ is solid, make it dashed. If $p(u)$ has a heavy child $h(p(u))$, make $(h(p(u)),p(u))$ solid.
\end{itemize}\noindent
We denote this collection of operations the \textit{heavy path operations}.

Splice and slice provide a natural way to implement expose and conceal for some vertex $u$. To expose, split $u$ and then splice every light vertex on the path from $u$-to-root. To conceal slice each light child on the $u$-to-root path and $u$. Note that all solid children are heavy when exposing. Using these, it is easy to see 

\begin{observation}
    Let $\ell(u)$ be the light depth of $u$. Then, on a proper tree, we can expose $u$ using $\ell(u)$ splices. On a $u$-exposed tree, we can conceal $u$ in $\ell(u)+1$ slices.
\end{observation}

\begin{lemma}
    All the other operations change the colour of at most one edge.
\end{lemma}
\begin{proof}
To implement \V{link-exposed} and \V{cut-exposed}, we only need to add or remove a single edge. The argument for reverse is as follows:  If the tree is $u$-exposed, then \V{reversing} the root path does not change the colour of any $r\dots u$ edge, as the resulting tree is $r$-exposed. All the off-path edges are still dashed.
\end{proof}

\subsection{Implementing Heavy Path Decompositions}\label{hp:implement}

To represent heavy paths and dashed children of nodes, we shall use Bent, Sleator and Tarjan's Globally Biased Search Trees \cite{BentSleatorTarjan1985}. Let $P$ be an ordered set and $wt$ a positive weighting of $P$. 

\begin{theorem}\label{hp:fastpath}
\cite{BentSleatorTarjan1985}. There is a data structure that maintains a directed, vertex-weighted path $P$ as a binary search tree. To each internal node, the data structure associates a left and right subpath and the unique edge connecting the two. If $W$ is the total weight in $P$ then the following queries are supported:
\begin{itemize}
    \item \V{Path}$(v)$: For $v\in P$, find the root of the search tree containing $v$ in $\BigO(\log \frac{W}{wt(v)})$
    \item \V{Head}$(P)$,\V{Tail}$(P)$: Find the first and last vertices of $P$ in $\BigO(1)$ time.
\end{itemize}
and the following dynamic operations: 
\begin{itemize}
    \item \V{Reverse}$(P)$: Reverse the direction of $P$ in $\BigO(1)$ time.
    \item \V{Reweight}$(v, wt'_v)$: Change the value of $wt(v)$ to $wt'_v$ in a path with total weight $W$ before and $W'$ after. This is $\BigO(\max(\log \frac{W}{wt(v)}, \log \frac{W'}{wt_v'}))$. 
    \item \V{Concatenate}$(P,Q)$: Let $P$ be a path with $\V{Tail}(P)=u$ and $Q$ another path with $\V{Head}(Q)=v$. Insert the edge $(u,v)$. If $W_P$ and $W_Q$ are the weights of $P$ and $Q$ respectively then \V{Concatenate} is
    $$
    \BigO\left(\log \dfrac{\max(W_P,W_Q)}{\max(wt(u),wt(v))}\right).
    $$
    \item \V{Split}$(v)$: Delete the edges of $v$ resulting in up to three new paths. The lower part from \V{Head}$(v)$ up to but not including $v$, $v$ itself and $p(v)$ to $\V{Tail}(v)$
    . If $P=\pathh(v)$ and $W$ is the total weight on $P$ then split is $\BigO\left(\log \frac W {wt(v)}\right).$
\end{itemize}
If the items are ordered by their weight, the following operations are supported using \V{Concatenate} and \V{Split}.
\begin{itemize}
    \item \V{Insert}$(v,p)$: Inserts $v$ into path $p$ in $\BigO(\log \frac{W}{wt(v)})$.
    \item \V{Delete}$(v)$: Deletes $v$ from path $p$ in $\BigO(\log \frac{W}{wt(v)})$.
\end{itemize}
\end{theorem}

The goal is now to use \cref{hp:fastpath} to represent a heavy path decomposition of $T$. Direct the edges of $T$ towards the root. Then store the heavy paths of $T$ using a biased search tree, where each node is weighted by a to-be-defined weight function $wt$. Let each node store a pointer to its parent and a biased search tree of its children, weighted by $\w$, which we denote the child set. Using the parent pointers it is possible to traverse any $u$-to-root path using repeated \V{Tail} and \V{Path} calls. The heavy path weighting function $wt$ is defined as follows
$$
wt(v) = \begin{cases}
    \w(v) & \text{if $v$ has no solid child} \\
    \w(v)-\w(u) & \text{if $v$ has solid child $u$}.
\end{cases}
$$
The maintenance of this weighting is non-trivial but solvable using standard techniques. Refer to \cite{SleatorTarjan1983} for details. Note that the total weight in the child set of $v$ is $wt(v)$. We can now describe how to implement the heavy path operations. 
\begin{itemize}
    \item \V{splice}$(v)$: To splice a node $v$, first remove $v$ from the child set of $p(v)$. Then split $\pathh(p(v))$ by $p(v)$ and concatenate the upper part, $p(v)$ and $\pathh(v)$. Add the tail of the lower part to the child set of $p(v)$.
    
    \item \V{slice}$(v)$: To slice a node $v$, split $\pathh(v)$ by $p(v)$ and find $h(p(v))$ if it exists. Remove $h(p(v))$ from the child set of $p(v)$ then concatenate the upper part, $p(v)$ and the path of the heavy child. Add $v$ as a dashed child of $p(v)$. 
    \item \V{reverse}$(r)$: To reverse the root-to-$u$ path of a $u$-exposed tree use reverse on $\pathh(r)$. 
    \item \V{link-exposed}$(u,v)$: Add $v$ as a dashed child of $u$ and reweight $u$ according to the new weight.
    \item \V{cut-exposed}$(u)$: Split the root-path of $T$ at the edge connecting $u$ to $p(u)$. 
    
\end{itemize}

With this we see that the heavy path operations can be implemented using a logarithmic\footnote{To expose and conceal, we need $\BigO(\ell(v))=\BigO(\log n)$ splices and slices.} amount of search tree operations. We analyse the specific running times of each operation.
\begin{observation}
    For a path $P$, the total weight $W_P$ is $\w(\V{Tail}(P))$.
\end{observation}

\begin{lemma}
    The $v$-to-root path of any node $v$ has length $\BigO(\log \frac{W}{wt(v)})$.
\end{lemma}
\begin{proof}
    We call \V{Path} and \V{Tail} repeatedly to find the root. Let $v=p_1$ and let $p_2,p_3,\dots,p_{k-1}$ be the parents of each path tail and let $P_1,P_ ,\dots,P_k$ be the traversed heavy paths with $P_k$ containing the root. Note that \cref{hp:dashedepth} implies $k\leq1+\log \frac{W}{\w(v)}$. 
    
    Now by \cref{hp:fastpath} each pair of \V{Path} and \V{Tail} calls are $\BigO(\log\frac{W_{P_i}}{wt(p_i)})$ and since $\V{Tail}(P_i)$ is a dashed child of $p_i$ we observe that 
    $$
    W_{P_i} \leq wt(p_{i+1})
    $$
    Applying this inequality and summing the run times yields
    $$
    \BigO\left(\sum_{i=1}^{k-2}\log\frac{wt( p_{i+1})}{wt(p_i)} + \log\frac{W}{wt(p_{k-1})}\right)=\BigO\left(\log\frac{W}{wt(v)}\right)
    $$
    
\end{proof}

\begin{lemma}\label{spqr:spliceslice}
    \V{splice}$(v)$ and \V{slice}$(v)$ run in time $$
    \BigO\left(\log\left(\frac{\w(\rho)}{wt(p(v))}\right) + \log\left(\frac{wt(p(v))}{\w(v)}\right)\right) = \BigO\left(\log\left(\frac{\w(\rho)}{\w(v)}\right)\right)
    $$
     where $\rho$ is the tail of $\pathh(p(v))$.
\end{lemma}
\begin{proof}
We analyse splice first. Removing $v$ from the child list of $p(v)$ is done using delete. The total weight in the child list is at most $wt(p(v))$. Thus the delete takes $\BigO(\log\frac{wt(p(v))}{\w(v)})$. Splitting $\pathh(p(v))$ is $\BigO(\log \frac{\w(\rho)}{wt(p(v))})$. Concatenating $\pathh(v)$ to that of $p(v)$ is $\BigO(\log\frac{\w(\rho)}{wt(p(v))})$ as $p(v)$ is heavier than $v$. This changes $wt(p(v))$ and constitutes a reweight which is $\BigO(\log\frac{\w(\rho)}{wt(p(v))})$. Inserting the previously solid child into the child list of $p(v)$ is constant, as the child was heavy.
The time taken for slice is almost the same. Splitting and concatenating the relevant paths is identical. Finding and extracting the heavy child takes constant time. In the end, the only difference is adding $v$ as a child of $p(v)$. This is $\BigO(\log\frac{wt(p(v)}{\w(v)})$. 
\end{proof}

\begin{lemma}\label{hp:expose}
    Expose and conceal take time $\BigO(\log(\frac{W}{wt(v)}))$, where W is the total weight in $T$.
\end{lemma}
\begin{proof}
\input{tex/figures/expose}
    We begin with expose. As we have already shown the light depth of any node to be $\BigO(\log\frac{W}{\w(v)})$, we only need to show that the accumulated running times of splice and slice do not exceed this. Let $v$ be the node to expose and $r$ the root of $T$. Let $P_0=\pathh(v)$ and $P_1, P_2\,\dots P_k$ be the remaining solid paths on the root path and $\rho_i$ the tails of said paths. Refer to \autoref{fig:expose} for a graphical representation. 
    Expose will call splice on $\rho_0,\rho_1,\dots,\rho_{k-1}$ in sequence. Thus, the total time used in splice is 
    $$
    \BigO\left(\log\frac{\w(\rho_1)}{\w(\rho_0)} + \log\frac{\w(\rho_2)}{\w(\rho_1)} + \dots + \log\frac{\w(\rho_k)}{\w(\rho_{k-1})}\right) = \BigO\left(\log\frac{\w(\rho_k)}{\w(\rho_0)}\right)
    $$
    The final detail is the time to find the next light edge. Given $v$ we can find $P_0$ in $\BigO(\log\frac{\w(\rho_0)}{wt(v)})$ and given $\rho_{i-1}$ we can find $\rho_i$ in $\BigO(\log\frac{wt(p(\rho_{i-1}))}{\w(p_{i-1})} + \log\frac{\w(\rho_i)}{wt(p(\rho_{i-1}))})=\BigO(\log \frac{\w(\rho_i)}{\w(\rho_{i-1})})$. Thus, the traversal of the light edges is no slower than splicing and slicing on the respective nodes. The only exception is $v$, which must be taken into account. By noting that $\rho_k=r$, we find that expose is
    $$
    \BigO\left(\log\frac{W}{wt(v)}\right).
    $$
The analysis is analogous for conceal, but requires an additional detail to find light edges on the solid path see \cite{SleatorT85} for details.
\end{proof}

With expose and conceal, the correctness and time bounds of link, cut and reverse also follow.
\begin{theorem}
    There exists a data structure which maintains a weighted heavy path decomposition of a directed tree $T$ while supporting \V{expose}, \V{conceal}, \V{reverse}, \V{link-exposed} and \V{cut-exposed} in $\BigO(\log \frac{W}{wt(v)})$ time.
\end{theorem}\noindent

Since $wt$ is some internal weighting, the following bound is weaker but easier to conceptualise. Recall that $wt(v)\geq w(v)$. Then for any positive weighting, $w$, of the nodes in $T$ we get:
\begin{corollary}
        There exists a data structure which maintains a weighted heavy path decomposition of a directed tree $T$ while supporting \V{expose}, \V{conceal}, \V{reverse}, \V{link-exposed} and \V{cut-exposed} in $\BigO(\log \frac{W}{w(v)})$ time.
\end{corollary}

%% file: tex/figures/expose.tex
\begin{figure}
    \centering
\includegraphics[width=0.40\linewidth]{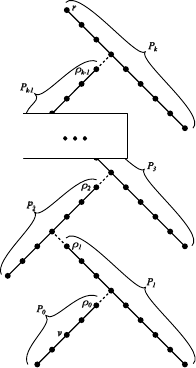}
    \caption{The layout of the paths towards the root when exposing some node $v$.}
    \label{fig:expose}
\end{figure}

%% file: tex/20incremental.tex
\section{Incremental Triconnectivity for Biconnected Graphs}\label{spqr}
In this section, we will construct a decomposition of biconnected graphs into a tree of their triconnected components. These were defined in their modern form by Hopcroft and Tarjan \cite{Hopcroft1973}. We will use a definition of Holm and Rotenberg \cite{HolmRotenberg2019} as it is easier to work with, however an easier to follow definition is given by Di Battista and Tamassia \cite{Battista1996} in their original work on SPQR-trees.

It is well known that all biconnected graphs can be constructed from $K_2^2$ by subdivision of edges and edge insertions (see e.g \cite{Whitney1932}) thus these dynamic operations allow us to span the entire set of biconnected graphs incrementally while supporting \V{Triconnected} queries. We can now move on to the data structure.

\subsection{SPQR-trees}

\input{tex/figures/spqr-example}

\begin{definition}\cite{DecrementalSPQR,HolmRotenberg2019}\label{spqr:props}
    The (strict) SPQR-tree $\Tree(G)$ for a biconnected graph $G$ with at least 3 edges is a tree with nodes labelled S, P and R. Each node $\mu$ has an associated \textit{skeleton graph} $\Gamma(\mu)$ with the following properties:
    \begin{enumerate}[label=(\roman*)]
        \item $V(\Gamma(\mu))\subseteq V(G)$.
        \item For any edge $e\in E(G)$ there is a unique node $\mu \in \Tree(G)$ with $e\in\Gamma(\mu)$
        \item If $x$ is an S-node, then $\Gamma(x)$ is a simple cycle with at least 3 edges
        \item If $x$ is a P-node, then $\Gamma(x)$ is a pair of vertices with at least 3 parallel edges.
        \item If $x$ is an R-node, then $\Gamma(x)$ is a simple triconnected graph.
        \item\label{spqr:partite} No two S-nodes are neighbours and no two P-nodes are neighbours.
        \item For every edge $e=(x,y)$ in the SPQR-tree, $V(\Gamma(x)) \cap V(\Gamma(y))$ is a separation pair $\{s,t\}$ in $G$ between the skeletons on either side of $(x,y)$. The virtual edge $(s,t)$ is in both $\Gamma(x)$ and $\Gamma(y)$. We say that $(x,y)$ is marked by $\{s,t\}$.  
        \item Any edge $e \in E(\Gamma(\mu))$ is either virtual or non-virtual. If an edge is non-virtual $e\in E(G)$.
    \end{enumerate}
\end{definition}
The SPQR-tree is unique (see e.g. \cite{Battista1996}). The skeletons of the tree nodes are also referred to as the \textit{triconnected components} of $G$. A \textit{relaxed} SPQR-tree is a tree that satisfies all but \cref{spqr:partite}.

\begin{lemma}
    \cite{Battista1996} For any SPQR-node $\mu$, $\Gamma(\mu)$ is homeomorphic to a subgraph of $G$. Also any 2-cut in $\Gamma(\mu)$  is a 2-cut in $G$.
\end{lemma}

\begin{lemma}\label{spqr:connected}
    The nodes containing a vertex $v$ form a connected component of $\Tree(G)$.
\end{lemma}
\begin{proof}
    Suppose two nodes of $\Tree(G)$ contain $v$ and don't form a connected component. Any edge separating them is marked by a 2-cut of $G$, this is a contradiction, as no vertex can be separated from itself.
\end{proof}
\begin{lemma}\label{spqr:twiceconnected}
    The nodes containing both $u$ and $v$ form a connected component of $\Tree(G)$.
\end{lemma}
\begin{proof}
    By the previous lemma, the nodes containing $u$ and $v$ respectively form two subtrees $\Tree_1$ and $\Tree_2$. Clearly the nodes containing both is $V(\Tree_1) \cap V(\Tree_2)$. It is well known that the intersection of two trees is connected.
\end{proof}

\begin{lemma}
    In any S or R-node, there is at most one virtual edge connecting $u$ and $v$
\end{lemma}
\begin{proof}
    Suppose there were two edges connecting $u$ and $v$. As the skeletons of either type are simple, such an edge would violate \cref{spqr:props}.
\end{proof}

\begin{lemma}\label{spqr:uniqueP}
    If more than two SPQR-nodes contain vertices $u$ and $v$, there is a unique P-node containing both.
\end{lemma}
\begin{proof}
    Suppose more than two nodes contain both $u$ and $v$. By \cref{spqr:twiceconnected}, these nodes form a connected component of $\Tree(G)$. The only nodes with degree greater than 1 must be P-nodes as these are the only nodes to admit more than 1 $u,v$ edge. Thus, all S and R-nodes must have a P-node as a neighbour. As no P-node neighbours another, this P-node must be unique for the component to be connected.
\end{proof}

\begin{observation}\label{spqr:star}
    If there are more than two SPQR edges marked with $u$ and $v$ they form a star around the unique P-node $\rho$ with $V(\Gamma(\rho))=\{u,v\}$.
\end{observation}\noindent
With this, we can define an important relation between pairs of vertices in $G$.

\begin{definition}\label{spqr:critpath}
    Given some SPQR-tree $\Tree(G)$ let $x,y\in V(G)$. The \textit{critical path}, $m(x,y)$ is then defined as: 
    \begin{enumerate}
        \item If exactly one SPQR-node, $\nu$, contains $x$ and $y$ then $m(x,y)=\nu$.
        \item If exactly two SPQR-nodes, $\mu,\nu$ contain both $x$ and $y$ then they are adjacent and $m(x,y)=\mu\ \nu$.
        \item If more than two SPQR-nodes containing both $x$ and $y$ then $m(x,y)$ is the unique P-node containing both.
        \item If no node contains both $x$ and $y$ then $m(x,y)$ is the unique shortest path $\mu \dots \nu $ in $\Tree(G)$ such that $x \in \Gamma(\mu)$ and $y \in \Gamma(\nu)$.
    \end{enumerate}
\end{definition}
Note that $m(x,y)$ exists and is unique \cite{HolmRotenberg2019}. We call a critical path of type 2. $u,v$-marked. See \cref{fig:spqr-critpath} for an example.

\input{tex/figures/spqr-critpath}

\subsection{Answering queries}

To answer the \V{Triconnected} queries, we will need the following lemmas, which explain how we can infer the triconnectedness of arbitrary nodes from the SPQR-tree.

\begin{lemma}\label{spqr:triconnected}
    A pair of vertices $u$ and $v$ are triconnected if and only if they share a P or R-node.
\end{lemma}
\begin{proof}
    $\Leftarrow$. Since the skeleton of any P or R-node is triconnected and since they are homeomorphic to a subgraph of $G$, any pair of vertices sharing a P or R-node are triconnected.
    $\Rightarrow$ If $u$ and $v$ are triconnected, they must share at least one node. 
    If there are two or more such nodes, one must be a P or R-node. Suppose then that $u$ and $v$ share a single S-node. The skeleton of this node is a cycle $C$ in $G$, on at least three vertices. Thus, at least one of the two paths on the cycle will contain a 1-cut. If the other does not contain a 1-cut, $u$ and $v$ are neighbours along $C$, and the remaining two paths connecting them will belong to a P-node neighbouring the S-node in question, contradicting our assumption. 
\end{proof}

\begin{observation}
    A pair of vertices $u,v\in G$ share a P or R-node if and only if $m(u,v)$ is $u,v$-marked or $m(u,v)$ is a single P or R-node.
\end{observation}

\subsection{Updates}\label{spqr:updates}
To insert edges, we first require two restructuring primitives: \V{split} and \V{merge}. Let $\mu_1,\mu_2\in\Tree(G)$ be adjacent nodes of the same type and let $\nu$ be an S-node of $\Tree(G)$. Then define
\begin{itemize}
    \item \V{merge}$(\mu_1, \mu_2)$: Contract the tree edge between $\mu_1$ and $\mu_2$ into a new node $\mu$. Let $e$ be the virtual edge representing $(\mu_1, \mu_2)$ in their respective skeleton graphs. Then the skeleton of $\mu$ is  $\Gamma(\mu_1)-e \cup \Gamma(\mu_2)-e$.
    \item \V{split}$(\nu)$: Pick two disjoint paths of $\nu$, $P_1$ and $P_2$. Remove the internal vertices of $P_1$ and $P_2$ from $\Gamma(\nu)$. This splits $\Gamma(\nu)$ into at most 2 disjoint paths. Connect the endpoints of $P_1$ and $P_2$ with virtual edges in $\Gamma(\nu)$ and create two new S-nodes $\nu_1$ and $\nu_2$ with $P_1$ and $P_2$ as their skeletons then add a virtual edge between their endpoints to form a cycle. The neighbours of $\nu$, $\nu_1$ and $\nu_2$ are exactly those corresponding to the virtual edges of their skeleton.
\end{itemize} 

The definition of \V{split} deviates from Di Battista and Tamassia \cite{Battista1996}, but is equivalent for paths of length 1 or 2. Splitting a node $\nu$ by a neighbouring SPQR-node corresponds to choosing a path equal to the virtual edge connecting it to $\nu$. Splitting by a vertex of $\nu$ means choosing one of the paths to be this vertex. An example of split can be seen in \cref{fig:splitmerge}.

With this, we can now describe how to insert an edge in the SPQR-tree. Let $(u,v)$ be an edge to be inserted. If some node contains both $u$ and $v$, the insertion turns out to be relatively simple. If not, we will need to look at $m(u,v)$, which shall be contracted into a single R-node if $u$ and $v$ are linked. That this is necessary follows from the following lemma.
\begin{lemma}[Implicit in \cite{Battista1996}]
    Let G be a graph with SPQR-tree $\Tree(G)$ and let $e=(u,v)$ be an edge not in $G$. Then all pairs of vertices contained in P- and R-nodes on $m(u,v)$ are triconnected in $G+e$. The connective properties of any vertex not contained in a node of this type are unchanged. 
\end{lemma} 
\noindent The procedure for inserting an edge into an SPQR tree is as follows.
\begin{itemize}
\item  \V{InsertEdge$(u,v)$}: Find $m(u,v)$ and let $e = (u,v)$.
\begin{enumerate}
    \item If $|m(u,v)|=1$ and it is a P or R-node simply insert $e$ into the skeleton. If the node is an S-node, then we split it between $u$ and $v$ and create a new P-node $\nu$ containing $e$. The two new S-nodes are both connected to the P-node.
    \item If $m(u,v)$ is $u,v$-marked subdivide the single edge of the critical path with a P-node containing $e$.
    \item If no node contains both $u$ and $v$ let $m(u,v) = \mu,\dots, \nu$ be their critical path.
    \begin{enumerate}[label=(\alph*)]
        \item Remove all SPQR-tree edges of $m(u,v)$. In P-nodes, remove the corresponding virtual edge as well. 
        \item Merge all R-nodes of $m(u,v)$ into a new node, $R^*$.
        \item Split all internal S-nodes between their neighbours. If an endpoint is an S-node, split it between its path neighbour and either $u$ or $v$. 
        \item Connect the new S-nodes and the old P-nodes to $R^*$.
        \item If the skeleton of any S- or P-node has $k<3$ edges, merge it into $R^*$.
    \end{enumerate}
\end{enumerate}
\end{itemize}\noindent
What remains now is to describe the subdivision of some edge $(u,v)$. 
\begin{itemize}
    \item \V{Subdivide}$(u,v)$. Let $w$ be the vertex to be inserted. We follow the cases of \cref{spqr:critpath}; note that the presence of $(u,v) \in E(G)$ means that cases 2 and 4 are irrelevant. Let $\chi$ be an S-node containing the cycle $u \ w \ v$ with $(u,v)$ being virtual. 
    \begin{enumerate}
        \item If $m(u,v)=\mu$ is a single node, that is not a P-node, let $(u,v)$ be virtual in $\mu$ and connect $\mu$ to $\chi$. If $\mu$ is an S-node, insert $w$ into the cycle instead.
        \item[3.] If $m(u,v)=\mu$ is a single P-node, connect $\mu$ to $\chi$ and make one of the non-virtual edges virtual.
    \end{enumerate}
\end{itemize}

%% file: tex/figures/spqr-example.tex
\begin{figure}
    \centering
    \includegraphics[width=\linewidth]{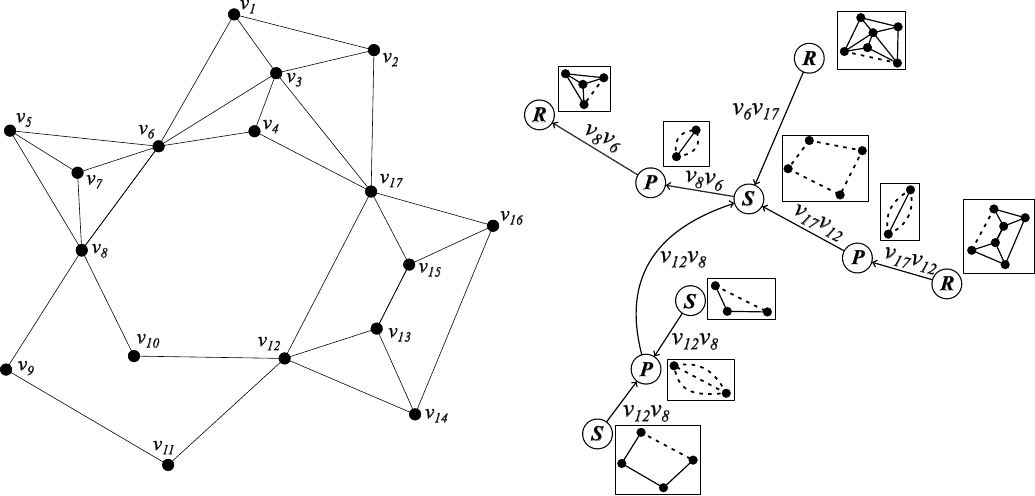}
    \caption{On the left a biconnected graph $G$. On the right the SPQR tree of $G$ with the skeletons pictured. Each edge is marked with the corresponding 2-cut. Note: Based on Figure 2 by Holm et. al. \cite{DecrementalSPQR}}
    \label{fig:spqr-example}
\end{figure}

%% file: tex/figures/spqr-critpath.tex
\begin{figure}
    \centering
    \includegraphics[width=0.5\linewidth]{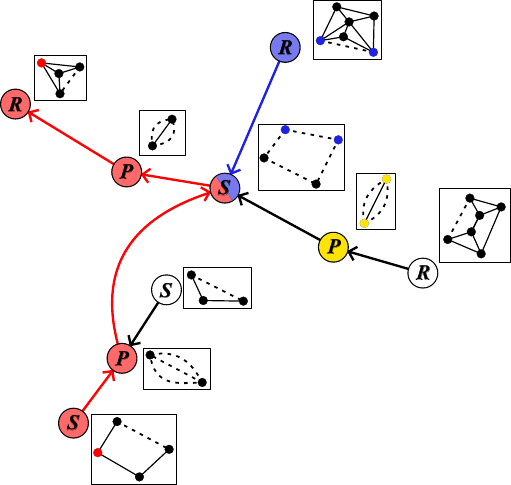}
    \caption{An SPQR tree with three different kinds of critical paths. Blue: Two vertices share exactly two nodes, the path is marked by the blue vertices. Yellow: Two vertices share more than two nodes. Red: Two vertices share no node. }
    \label{fig:spqr-critpath}
\end{figure}

%% file: tex/30IncrementalWorstCase.tex
\section{SPQR-trees as heavy path decompositions}\label{spqr2}

In this section, we will use a heavy path decomposition to implement SPQR-trees with worst-case guarantees. The main challenge in this regard will be to prevent \V{InsertEdge} from requiring a traversal of $m(u,v)$. We will do this by reducing edge insertions to a constant amount of heavy path operations and a contract. This reduction is due to Holm and Rotenberg \cite{HolmRotenberg2019} .

Let $\Tree(G)$ be an SPQR-tree represented by a heavy path decomposition with the nodes marked S and P. With each edge we store its marking and a non-negative value. We will represent R-nodes as connected subtrees of nodes connected by edges with non-zero value. A node in this setting is then either a normal S or P-node or a node with an adjacent, positively valued edge, indicating that said node is part of the representation of an R-node.
The representation of S and P nodes closely resembles the representation discussed in the previous section; however we will defer the details to \cref{spqr:repr}.
For now we shall first define a variant of the SPQR-tree. 

\subsection{Splitting along heavy paths}
Recall that during the insertion of an edge, we split each S-node of the critical path by its neighbours and remove up to 3 edges from every P-node. In this section, we will show how to do this during other operations so that relevant nodes are already split when the time comes to insert an edge. 
We will first formally define how exactly to split each node and then give a mechanism to pre-split all nodes required by any insertion. The choice will essentially follow the cases from the definition of \V{InsertEdge} from \cref{spqr} where we let the heavy paths of $\Tree(G)$ stand in for the critical path.

\begin{definition}
    Let $x$ and $y$ be vertices of $G$. Let $\Tree(G)$ be a rooted relaxed SPQR-tree of $G$ represented by a heavy path decomposition. We can \textit{expose} a vertex $x$ of $G$ by exposing some node of $\Tree(G)$ containing $x$. This makes $\Tree(G)$ $x$-exposed. Similarly, we can root $\Tree(G)$ at $y$ if some node containing $y$ is the root of $\Tree(G)$. We say that $\Tree(G)$ is rooted at $y$. We refer to $x$ and $y$ as the selected vertices of $G$.
\end{definition}

\begin{definition}
    A node $u$ in $\Tree(G)$ is \textit{primary} if it is the root or has a parent of a different type. Otherwise $u$ is \textit{secondary}.
\end{definition}

\begin{definition}\label{incworst:childgroup}
    Let $\mu$ be a primary node in a relaxed SPQR tree represented by a heavy path decomposition with the vertex $x$ exposed and rooted at $y$. Let the solid degree $s(\mu)$ be the number of solid edges connected to $\mu$. Then we define up to 2 \textit{child groups} of $\mu$ as disjoint subsets of the edges of $\Gamma(\mu)$.
    \begin{enumerate}
        \item If $\mu$ is an S-node, we have four possible cases:
        \begin{enumerate}
            \item If $\mu$ has 2 solid edges $e_p$ and $e_c$, then the child groups are the (possibly empty) sets of edges between the virtual edges of $\Gamma(\mu)$ corresponding to $e_p$ and $e_c$.
            \item If $\mu$ has $s(\mu)=1$, is an endpoint of the solid path of the root, and contains either $x$ or $y$, then the child groups are the edges on either side of the virtual edge corresponding to the solid edge and the aforementioned vertex. 
            \item If $\mu$ is the tail of a non-root solid path and has $s(\mu)=1$, then its child groups are the sets of edges between the virtual edges of its parent and solid child.
            \item If $\mu$ has $s(\mu)=0$, is the root and contains both $x$ and $y$, then the child groups are the edges on either side of $x$ and $y$. 
        \end{enumerate}
        \item If $\mu$ is a P-node on the root path containing a selected vertex or a P-node with $s(\mu)=2$; it has at most 1 child group consisting of all edges that don't correspond to an incident solid edge.
        \item In all other cases, there is no child group.
    \end{enumerate}
    A proper child group is a child group with more than 1 edge. 
\end{definition}

Notice that the child groups of an internal node on a heavy path are exactly the virtual edges that would not be included in $R^*$ when inserting an edge between vertices on either side of the node. We shall now define how we separate and unite the child groups from their primary node by redefining split and merge in the context of heavy paths.
\begin{definition}
    Let $\mu$ be a node of a relaxed SPQR-tree with a heavy path decomposition. We have two operations:
    \begin{itemize}
        \item \V{Split$(\mu)$}: Separate the proper child groups of $\mu$ from $\Gamma(\mu)$ and create at most 2 new child nodes of $\mu$ with the proper child groups as skeletons. If any virtual edges are moved, move their corresponding tree-edges to the new children. Connect the new children to $\mu$ with a dashed edge (and insert a virtual edge). If the new children have a heavy child, make the edge between them solid. 
        \item \V{Merge$(\mu)$}: Contract the edges of the at most 2 secondary children of $\mu$ into a single node. If any secondary child had a solid edge, make it dashed. 
    \end{itemize}
    \cref{fig:splitmerge} shows an example of the two operations.
\end{definition}

\input{tex/figures/splitmerge}

With a way to operate on child groups, we can define which nodes are split at what times and how we maintain this under the heavy path tree operations.

\begin{definition}
    A node is \textit{splittable} if it has a proper child group. A node is \textit{mergeable} if it has a secondary child.
\end{definition}

\begin{lemma}\label{split:invariants1}
    Given a relaxed $x$-exposed SPQR-tree rooted at $y$, the following invariants hold under splits and merges:
\begin{enumerate}[series=invariants,label=(I-\arabic*)]
    \item No node is both splittable and mergeable. \label{inv1}
    \item Every secondary node is the dashed child of a primary node, and its skeleton is a proper child group of its parent. Every mergeable primary node has at most 1 solid child and at most 2 secondary children.\label{inv2}
\end{enumerate}
Furthermore, any such tree can be constructed from a strict SPQR-tree of the same graph by splitting a subset of the splittable nodes.
\end{lemma}
\begin{proof}
In a strict SPQR-tree \cref{inv2} holds trivially. \cref{inv1} holds as no node is secondary and thus no node is mergeable.

Given a relaxed SPQR tree where the invariants hold and a splittable primary node $\mu$, splitting $\mu$ makes it mergeable. It is not splittable since all proper child groups were removed by the split. After the split, all edges in a proper child group have been moved to a secondary child, thus ensuring \cref{inv2} for $\mu$. The parent of $\mu$ is unchanged, as is its solid child. $\mu$'s dashed children have all changed parents, but as $\mu$ was splittable, they were all primary before the split and thus cause no violations of \cref{inv2}. All other nodes are unchanged. Thus split preserves the invariants.

Suppose now that $\mu$ is a mergeable node. By \cref{inv2} $\mu$ is primary and has at most 2 secondary children. After a call to merge, $\mu$ is no longer mergeable; instead, the skeleton of $\mu$ now contains child groups, making $\mu$ primary and splittable. All children of $\mu$ after the merge were by \cref{inv2} primary. Merge does not change this. All other nodes are unchanged. Thus, merge preserves the invariants.
To go from a strict SPQR tree to a given relaxed one, use split on a subset of the splittable nodes. 
\end{proof}

With this, we see that every strict SPQR-node is represented by at most 3 nodes in our relaxed SPQR-tree, maintaining the linear space complexity. To use our new definition of child groups, we must also be able to freely choose the selected vertices $x$ and $y$. We introduce two additional operations:
\begin{itemize}
    \item \V{Select}$(u,v)$: Given a proper heavy path decomposition of $\Tree(G)$ and vertices $u,v \in V(G)$. Find $\mu$ containing $u$ and $\nu$ containing $v$. First expose $\mu$ and reverse to root $\Tree(G)$ at $u$. Then conceal the old root and expose $\nu$. Set $\Tree(G)$ to be rooted at $u$ and $v$-exposed.
    
    \item \V{Deselect}: Given a heavy path decomposition of $\Tree(G)$ rooted at $u$. Return a heavy path decomposition without any selected vertices.
\end{itemize}

Note that these two operations might break \cref{inv2}, as they might change the child groups of some nodes. Before fixing this, we shall first define the final form of our SPQR-tree:

\begin{lemma}
    Given a $u$-exposed relaxed SPQR-tree $\Tree(G)$ rooted at $v$ satisfying the invariants of  \cref{split:invariants1}. There is a corresponding relaxed SPQR-tree $\Tree'(G)$ that satisfies the same invariants in addition to the following:
    \begin{enumerate}[resume=invariants,label=(I-\arabic*)]
        \item\label{inv3} A primary node is only splittable if it is the dashed child of an R-node.
    \end{enumerate}
    We call $\Tree'(G)$ a pre-split SPQR tree of $G$ with selected vertices $u$ and $v$ and the invariants the \textit{pre-split invariants}.
\end{lemma}
\begin{proof}
    Starting with $\Tree(G)$, we obtain the pre-split SPQR tree of $G$ by splitting all the splittable primary nodes. By \cref{split:invariants1} this preserves \cref{inv1} and \cref{inv2} and establishes \cref{inv3} for all violating primary nodes.
\end{proof}

Furthermore, we can also make our new operations respect the new invariant.
\begin{lemma}\label{presplit:mergesplitinvariants}
    We can restore the pre-split invariants after split or merge using a constant amount of additional splits and merges.
\end{lemma}
\begin{proof}
    Let $\mu$ be some primary SPQR-node to be split. If $\mu$ is a P-node, then splitting it affects the heavy path of at most one child. However the endpoint of this path must be an S or R-node and thus in either case, the child groups do not change when adding something to the end of the path. Thus a split of $\mu$ preserves \cref{inv3}. If $\mu$ is an S-node then splitting it affects the heavy path of at most two children. If any of these children were a heavy P-node, the child must be split as the P-node is now internal to the heavy path and might now be splittable. This split however does not cascade further as P-nodes can only have S or R-node children. Thus a split of some S-node also preserves \cref{inv3}.

    After a merge of some node $\mu$, any of its former secondary children might have had a solid child which is now the endpoint of its solid path. If this is an S-node, its child groups do not change. If this is a P-node, it might necessitate an additional merge, which stops as its children are either S or R-nodes. Thus \cref{inv3} is preserved by merge. 
\end{proof}

\begin{lemma}
    We can restore the pre-split invariants after a select or deselect using a constant number of merges or splits.
\end{lemma}
\begin{proof}
    After select or deselect the tree satisfies all but \cref{inv2}. Specifically, the endpoints of the root solid path might violate the child group conditions. A merge followed by a split at each end restores the invariants. To see that this is indeed the case, consider the fact that only the child groups of the root path are affected by the choice of selected vertices. Additionally, all internal nodes on the path have solid degree 2 and are thus also unaffected by the choice. The only remaining violators are the endpoints.
\end{proof}

With this, we have defined exactly which nodes should be split at what times. What we now need is to also maintain this while the SPQR-tree changes.

\subsection{Maintaining a Pre-Split SPQR-tree Under Heavy Path Operations}

This section will show how to maintain a pre-split SPQR tree under all heavy path operations. These will later allow us to implement \V{InsertEdge}, \V{Subdivide} and \V{Triconnected}. Recall that expose and conceal only require splice and slice. However, these currently violate the invariants. Suppose that $\Tree(G)$ has no selected vertices. We extend them as follows: 
\begin{itemize}
    \item \V{Splice}$(\mu)$: Make $p(\mu)$ primary and splittable by merging $p(\mu)$ or $p(p(\mu))$ and let $\mu'$ be the (new) tail of $\mu$'s solid path. Execute splice as normal on $\mu'$, split $p(\mu)$, $p(p(\mu))$ and $\mu'$ and merge the formerly solid heavy child of $p(\mu)$ if relevant.
    \item \V{Slice}$(\mu)$: Merge $p(\mu)$ and $\mu$ then slice $\mu$ as normal and restore the pre-split invariants by splitting $p(\mu)$ and $\mu$. 
\end{itemize}

\begin{lemma}\label{spqr:spliceslice}
    Splice and slice preserve the pre-split invariants. Assuming that split and merge take worst-case $\BigO(1)$ time, then splice and slice run in $\BigO(\log\frac{W}{\w(\mu')})$ time where $W$ is the weight of the solid path of $p(\mu)$. This does not change the analysis of expose and conceal.
\end{lemma}
\begin{proof}
We first consider splice. After the merge of $p(\mu)$, the pre-split invariants hold in all but $p(\mu)$ and $\mu'$. Splicing then invalidates them for $p(\mu)$'s solid child. Splitting $\mu'$, $p(\mu)$ and merging $p(\mu)$'s former solid child fixes this.

For slice, the merge of $p(\mu)$ and $\mu$ violates the pre-split invariants only in $p(\mu)$ and $\mu$. After slicing $p(\mu)$ might be splittable. This is not the case for $\mu$ as it is the endpoint of a non-root solid path. The re-split fixes this. $p(\mu)$'s potential heavy child satisfies the pre-split invariants at all times as it must also be heavy as a child of a secondary node.

The number of merges and splits is clearly bounded, thus the worst-case bound for both operations is unchanged. The analysis of expose and conceal still works; however, we need to argue that splicing $\mu'$ instead of $\mu$ does not change the run time. This is simple since $\mu'$ is either $\mu$ itself or a heavy child. Thus $2\w(\mu')>\w(\mu)$ which implies that the running time of \V{splice} is still $\BigO(\log\frac{\w(\rho)}{\w(\mu)})$ where $\rho$ is the tail of $p(\mu)$'s heavy path.
\end{proof}\noindent
We will defer the analysis of merge and split to the next section. For now, this result immediately implies the following 

\begin{corollary}
    Expose and conceal preserve the pre-split invariants and run in$\BigO(\log \frac{W}{wt(\mu)})$ time.
\end{corollary}

Now that we can freely select and expose vertices and nodes, we can move on to modifying the SPQR-tree using the other heavy path operations. 

\begin{lemma}
    Given a $\mu$-exposed pre-split SPQR tree $\Tree(G_1)$, another pre-split SPQR tree $\Tree(G_2)$ with root node $\nu$, both without any selected vertices and supposing that $\mu$ and $\nu$ share exactly two vertices which constitute a cut of $G_1 \cup G_2$ and are of different types. Then connecting $\mu$ and $\nu$ with a dashed edge doesn't break any invariants. Additionally, this insertion can be done using link-exposed$(\mu,\nu)$.
\end{lemma}
\begin{proof}
    As no vertices are selected in either of the SPQR-trees, we know that $\mu$ and $\nu$ have no child groups and thus the insertion of a dashed edge does not break any invariants in either. 
\end{proof}
\begin{lemma}
    Given a pre-split SPQR-tree $\Tree(G)$ without selected vertices. Deleting the edge $(\mu,p(\mu))$ does not violate the pre-split invariants and can be done using cut-exposed$(\mu)$ and a constant number of merges.
\end{lemma}
\begin{proof} 
     By first exposing $\mu$, then using cut-exposed on $\mu$, we delete the edge in question, but now the solid degree of $p(\mu)$ has changed and thus \cref{inv2} and \cref{inv3} might be violated. By merging $p(\mu)$ we restore \cref{inv2} and \cref{inv3}. The solid degree of all other nodes is unchanged, and thus the other invariants are preserved. 
\end{proof}

\begin{lemma}
    In a pre-split SPQR tree with root $r$ and no selected vertices, reverse$(r)$ does not violate the pre-split invariants.
\end{lemma}
\begin{proof}
    \cref{inv1} is unaffected, as reversing a path does not affect the child groups. By \cref{inv2}, every node along the root path is primary, and thus the first part of \cref{inv2} is also preserved. Every internal node of the root path both gains and loses a solid child. The root loses a solid child, and the other endpoint of the root path gains one. In all cases, the number of solid children is the same after the reversal. The number of secondary children of all nodes is also unchanged. Thus, we preserve \cref{inv2}. The splittability of the internal nodes is unaffected. As no nodes are selected, this is also the case for the endpoints. Thus \cref{inv3} is preserved.
\end{proof}

Before explaining how to insert edges in the SPQR-tree, we must first introduce a new operation on the heavy path decomposition.
\begin{itemize}
    \item \V{contract}$(r)$: Where $r$ is the root of a heavy path decomposition. Contract the heavy path containing the root into a single R-node.
\end{itemize}

Due to our choice of representation, the contraction of some heavy path amounts to adding $1$ to the value of every edge on the path. This operation can be supported by lazily propagating a token down the search tree storing the heavy path see e.g. \cite{SleatorTarjan1983, top1}. 

\begin{observation}
    \V{contract}$(r)$ can be implemented in $\BigO(1)$ time. 
\end{observation}

As contract makes every P and S-node along the heavy root path into a single R-node, we might have invariant violations at each dashed child. In particular all the secondary children of the path are now primary.

\begin{lemma}
    \V{contract} does not violate the pre-split invariants.
\end{lemma}
\begin{proof}
    All secondary children of primary nodes on the root path are now themselves primary nodes, which are children of the new R-node. This does not cause any violations of \cref{inv3}. \cref{inv1,inv2} are both preserved trivially.
\end{proof}

\subsection{Representing the Pre-split SPQR-tree}\label{spqr:repr}
With the technical description of the data structure done, we can describe the details of representing the nodes of the pre-split SPQR-tree. 

As discussed in \cref{hp:implement}, a node of any heavy path decomposition needs to store at least its parent edge and a data structure which keeps track of a set of dashed children. The child set should support fast retrieval of heavy children and insertion/deletion in $\BigO(\log\frac{W}{\w(c)})$ for an instance with total weight $W$ and child $c$. A P-node has no inherent ordering of the edges and as such, the 'default' representation by way of a weight-ordered biased search tree is sufficient for storing the child set. An S-node on the other hand must keep track of the cyclical ordering of the dashed children imposed by the virtual edges of the skeleton. Let $e$ be the solid child of some S-node $s$. $e$ and $(s,p(s))$ divide the virtual edges of the skeleton into two paths. We store these in two separate biased binary search trees in the correct cyclical ordering, splitting and concatenating as needed when $s$ changes solid children. In particular, when making a solid edge dashed we merge the two child lists and the solid child into a single child set. When making a dashed edge solid, we split the child set by said edge. This takes constant time if the edge is heavy. Additionally we have
\begin{lemma}
    We can determine if the child set of any node contains positive or 0-valued edges in constant time along with the heaviest dashed child in constant time.
\end{lemma}
\begin{proof}
    The child set consists of at most 2 biased search trees over the edges. By letting every internal node of the search tree keep a bit indicating the presence of a 0-edge and a positive edge in its subtree we can maintain the required information in the root. The bit of any internal node is easily maintained during \V{Insert} and \V{Delete}. The same technique allows tracking of the heaviest child.
\end{proof}

\begin{lemma}
    We can maintain the child groups of S and P-nodes under \V{splice}, \V{slice}, \V{reverse}, \V{link-exposed} and \V{cut-exposed} without additional asymptotic cost.
\end{lemma}
\begin{proof}
    The representation of P-nodes is unchanged, hence we only consider S-nodes. When splicing and slicing the node $\mu$, its dashed and solid children don't change. This is not the case for $p(\mu)$. Suppose that we are splicing $\mu$ and that $p(\mu)$ has a heavy child $h(p(\mu))$. Then the required 3-way concatenate between the left child set, $h(p(\mu))$  and the right child set is
    $$
    \BigO(\log \frac{wt(p(\mu))}{\w(h(p(\mu)))}).
    $$ 
    Afterwards split this by $(\mu,p(\mu))$ in 
    $$
    \BigO(\log \frac{wt(p(\mu))}{\w(\mu)}).
    $$
    If $p(\mu)$ does not have a heavy child, forego the first step. This is clearly an 
    $$
    \BigO(\log\frac{wt(p(\mu))}{\w(\mu)})
    $$ extension to \V{splice}.
    The situation is symmetric for slice.
    To implement \V{reverse}, we only need to describe how to reflect the propagation of the lazy reverse bit. This swaps the parent and heavy child of $\mu$. The child lists of $\mu$, however, already reflect the split of the dashed edges between these; only the order needs to be swapped. We do this by swapping the two lists and reversing both. \V{link-exposed} and \V{cut-exposed} simply add or remove a dashed child from the child list. Doing this might require a split followed by a merge, however, only once. This is no slower than linking and cutting in the previous representation, thus the run time is unchanged.
\end{proof}\noindent
Now to enable split and merge we also need one addition to our data structure. Let each non-root solid path have an optional tail variable where a new tail can be added without a change to the biased binary search tree. With this, we can analyse merge and split.
\begin{lemma}
    Merge and split can be implemented in $\BigO(1)$ time.
\end{lemma}
\begin{proof}
    As the child lists are maintained in splice and slice, merging and splitting become easy. 
    To split a P or S-node $\mu$, split off each side of the child lists and assign them to new nodes. This also implicitly changes the parent pointer of all dashed children, as they point to an item in the dashed child list of $\mu$.
    %
\end{proof}
The inclusion of these lazy tails do not change the implementation of \V{reverse}, \V{link-exposed} and \V{cut-exposed} as these operate on the root solid path. We now prove that we can make \V{splice} and \V{slice} (and consequently \V{expose} and \V{conceal}) respect lazy tails.
\begin{lemma}
    The tail of a solid path is lazy only if it is secondary or it is the child of an R-node.
\end{lemma}
\begin{proof}
First note that this invariant is preserved by any call to split and merge, as a primary node can be split and merged at most once. The heavy path operations all merge and split invalidated nodes before and after operating thus preserving the invariant. The only exception is \V{contract} which gives every edge on its path a positive value, making all secondary children of the root path primary; which is allowed as they are all now children of an R-node. Since no two secondary nodes are neighbours, this also bounds the number of lazy endpoints of any given path to 1.
\end{proof}
\begin{lemma}
    Let a solid path with a lazy endpoint be dirty and one without clean. We can modify the biased search tree data structure of Bent, Sleator and Tarjan \cite{BentSleatorTarjan1985} to lazy tails with the following operations:
    \begin{itemize}
        \item \V{AddLazy}$(P,i)$: For clean $P$, add $i$ as a lazy endpoint of $P$ in $\BigO(1)$ time. 
        \item \V{RemoveLazy}$(P)$: Remove and the return the lazy endpoint of $P$ in $\BigO(1)$ time. This leaves $P$ dirty.
        \item \V{ConcatenateLazy}$(P, i, Q)$: Where $Q$ is clean. Return the concatenation of $P$, $i$ and $Q$ in worst-case $\BigO(\log\frac{\max(W_P,W_Q)}{\max(w_{p^+},w_i,w_{q^-})})$. 
        \item \V{SplitLazy}$(P,i)$: Split $P$ at $i$, yielding two new trees, one consisting of the upper part of $P$ and $i$ and one consisting of the lower part in $\BigO(\log\frac{W_P}{w_i})$.
    \end{itemize}
    Using these instead of their non-lazy counterparts in \V{splice}, \V{slice}, \V{expose} and \V{conceal} does not change the running time. 
\end{lemma}
\begin{proof}
    To support adding and removing a lazy endpoint, we will simply store it in the root of the biased search tree as needed. To support the concatenation of $P$ and $Q$, with some item $i$ in the middle we do the following: If $w_i \geq w_{p^+},w_{q^-}$, concatenate $P$, $i$ and $Q$ as normal. Otherwise start concatenating $P$ and $Q$ as normal. At some point during this operation we touch the lowest common ancestor of $p^+$ and $q^-$ in the search tree. Store $i$ here as a \textit{lazy item}\footnote{If an internal node has degree $B$, it can be the lowest common ancestor of at most $B-1$ leaf pairs, thus its out-degree at most doubles.}. We maintain that between two non-lazy items, there is at most one lazy item. This also maintains the invariant that for any lazy item,  at least one neighbour has strictly larger weight. \\
    To split at some non-lazy item $i$ use the normal split procedure, and associate any lower lazy neighbour of $i$ to the lower half. This is possible as split traverses the entire $i$-to-root path. If $i$ is lazy, split at $i$'s heaviest neighbour and add $i$ as a non-lazy child of the upper half using \V{Insert}. \\
    During  \V{splice} and \V{slice} we can assume cf. \cref{spqr:spliceslice} that no participating node is secondary and that we can restore the pre-split invariants after the slice in constant time. To \V{splice} $\mu$, first split $p(\mu)$'s solid path at $p(\mu)$. Then insert the  lower (possibly dirty) half as a dashed child of $p(\mu)$. Then if $\mu$ is not a lazy endpoint, we concatenate its heavy path as normal, otherwise, it is the lazy endpoint of an otherwise clean path and we use \V{ConcatenateLazy} to make $(\mu,p(\mu))$ solid. Note that even if $p(\mu)$ was lazy before the split, it is not lazy after.

    To \V{slice} $\mu$ split $p(\mu)$'s solid path at $p(\mu)$, remove $p(\mu)$'s heavy child from its child set and concatenate their solid paths. This may leave $\mu$ and heavy child of $p(\mu)$ as lazy. The running time is the same as for \V{splice}.
    As in \cref{hp:expose} this implies the running times for \V{expose} and \V{conceal}. 
\end{proof}

We can then move on to the representation of R-nodes. As mentioned previously, these are represented by connected components of edges with positive value. Fortunately, the pre-splitting system does not require R-nodes to transform during splice and slice. However, there are still implementation details to discuss. We will denote an S or P-node with a positive edge as \textit{obsolete}. The representation of a relaxed SPQR-tree will then only consist of S-nodes, P-nodes, and obsolete nodes. Note that since \V{contract} only operates on solid paths, the number of vertices in the skeleton of any obsolete node is at most 4. As an R-node is not represented by a single object, we will let each obsolete node represent the R-node as a whole and thus when we say e.g. \V{expose} an R-node, we will simply expose any obsolete node representing the R-node. 

To attach a new adjacent edge to the R-node when e.g. using \V{Subdivide}, we do the following: For each vertex $v$ of $G$, pick consistently a node of $\Tree(G)$ whose skeleton contains the vertex; we call this node the \textit{representative} of $v$. This is stored in a simple deterministic logarithmic-access-time dictionary over all vertices, and the representative is picked in an arbitrary but consistent way, say e.g. as the first node to contain the vertex\footnote{Note that in order to facilitate \V{select}, any S-node representative must also point to the location in their child set which represents their representee vertex.}. Consider then the representatives $\mu$ and $\nu$ of the nodes $u$ and $v$ and consider the path $\mu ~ \dots ~ \nu$. Let $\nu^\disk$ be the node containing $v$ closest to $\mu$ and define $\mu^\disk$ symmetrically. We call $\mu^\disk$ and $\nu^\disk$ the $(u,v)$-poles of $\Tree(G)$. Note that the path $\mu^\disk \ \dots \ \nu^\disk$ must intersect the critical path $m(u,v)$ for any vertex pair. Note also that through varying potential choices of representative for $u$ and $v$, if $u$ and $v$ do not share an SPQR-node, then the $(u,v)$-poles remain unique. 

Now, the 0-valued neighbours of any R-node will either be left over from a contraction in the form of former secondary children attached to some previously split obsolete node or attached after the initial contraction. Observe first that
\begin{observation}
    No SPQR-tree operation subdivides a positive edge.
\end{observation}
Thus, once both poles are contained in the same R-node, they never change again. Consequently, if we are attaching a new SPQR edge to an R-node, then it must be the only node to contain the marking of the new edge. Thus, the pole-path must be entirely contained in an R-node which means that its endpoints form a natural placement of the new edge in the R-node. The child set of an obsolete node can thus be characterised as the at most 4 original edges from before the contraction, plus any edges added later, where the node is a pole of the marking.
We have now seen that all SPQR-tree operations will rely on the ability to expose the critical path, which we will now show how to find given this representation.

\input{tex/97FindCritPath}

\input{tex/HPQueryAndOperation}

%% file: tex/figures/splitmerge.tex
\begin{figure}[t]
    \centering
    \begin{subfigure}{0.65\linewidth}
        \centering
        \includegraphics[width=\linewidth]{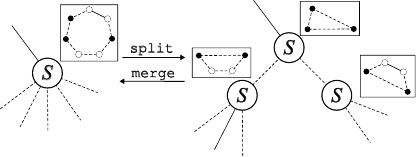}
        \caption{Splitting an S-node $s$ with a single selected vertex. The virtual edge with solid endpoints corresponds to the solid edge of $s$. After the split $s$ has two secondary children, one of which happens to have a heavy child.}
    \end{subfigure}

    \begin{subfigure}{0.65\linewidth}
        \centering
        \includegraphics[width=\linewidth]{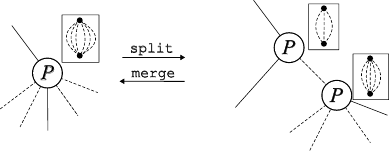}
        \caption{Splitting a P-node $p$ with 2 solid edges. After the split $p$ has a single secondary child containing all dashed children of $p$. One dashed child is heavy and is thus a solid child in the secondary node.}
    \end{subfigure}
    \caption{The result of split and merge for the different kinds of SPQR-nodes. The skeletons are drawn immediately next to the nodes. Virtual edges are marked with dashed lines, and non-virtual marked with solid lines.}
    \label{fig:splitmerge}
\end{figure}

%% file: tex/97FindCritPath.tex
\subsection{Finding critical paths}
Now to find the critical path of a pair of vertices, we can use the fact that the $(u,v)$-poles intersect the path. In particular, if there are nodes shared by $u$ and $v$, they are guaranteed to be poles, and if they do not share a node, the poles are exactly the endpoints of the critical path. We will use a well-known technique from dynamic trees and use the search tree representation of solid paths and dashed children. As explained in \cite{SleatorTarjan1983}, we can view each biased search tree as a binary tree over the edges. In particular we can view the biased search tree over a path as a binary tree where each internal node represents the paths of its right and left children, connected to an edge which resides in the internal node. By storing auxiliary information in each internal node, such as e.g. the minimum edge on the path, we allow logarithmic time queries about entire paths. 

Before computing the actual critical path of some pair of nodes, we first need some definitions.
\begin{definition}
    Let $\pi$ be the root path of an SPQR-tree $\Tree(G)$. We define the following:
    \begin{itemize}
        \item $\Z(\pi)$: The number of 0-valued edges on $\pi$.       
        \item $U_u(\pi)$: The maximal subpath of $\pi$ starting at the upper endpoint of $\pi$ using only edges marked by $u$
        \item $L_u(\pi)$: The maximal subpath of $\pi$ starting at the lower endpoint of $\pi$ using only edges marked by $u$.
    \end{itemize}

\end{definition}\noindent
Note that for any path $P$, $U_x(P)$ and $L_x(P)$ are empty for all except two vertices of $G$. Note also that positively valued edges are not truly edges of the SPQR-tree and, in a sense, act as edges with every marking. In the following we will only store $U_x$ and $L_x$ for any path using the marking of the first edge, but let any subsequent positive edge continue the path.
\begin{lemma}\label{spqr:pathinfo}
    It is possible to extend the data structure of \cref{hp:fastpath} to, in each internal node $I$ of the biased search tree, store the following information about the subpath $\pi(I)$ it represents:
    \begin{itemize}[label=$\bullet$]
        \item $|\pi(I)|$, $\Z(\pi(I))$
        \item For $U_x(\pi(I))$ of non-zero length: $|U_x(\pi(I))|$ and $\Z(U_x(\pi(I)))$ along with the lowest $x$ marked edge.
        \item For $L_x(\pi(I))$ of non-zero length: $|L_x(\pi(I))|$ and $\Z(L_x(\pi(I)))$ along with the highest $x$ marked edge.
    \end{itemize}
    This information can also be maintained while adding a positive $\Delta$ to the entire path.
\end{lemma}\noindent
\begin{proof}
If the node is a leaf it does not represent a subpath and the information is trivial to compute. If the node is internal with upper child $X$, lower child $Y$ and assigned edge $e$ then
$$
|\pi(I)| = |\pi(X)| + |\pi(Y)| + 1.
$$
and 
$$
\Z(\pi(I)) = \begin{cases}
    \Z(\pi(X)) + \Z(\pi(Y)) + 1 & \text{$e$ has value 0}  \\
    \Z(\pi(X)) + \Z(\pi(Y)) & \text{otherwise.}
\end{cases}
$$
If $e$ is marked by $x$ or positive and $|U_x(\pi(X))|=|\pi(X)|$ then
$$
U_x(\pi(I))=U_x(\pi(X))\cup e \cup U_x(\pi(Y))
$$
otherwise
$$
U_x(\pi(I))=U_x(\pi(X))
$$
with this $|U_x(\pi(I))|$ and $\Z(U_x(\pi(I))$ are easily computed along with the highest $x$ marked edge. To maintain this information under reverse, it is enough to simply swap $U_x$ and $L_x$. When propagating a positive $\Delta$, the entire path now becomes positive and thus the number of zeroes on the paths becomes 0. To compute $U$ and $L$, we can simply keep a version where we pretend that the entire path is positive alongside the normal values and swap when propagating. 
\end{proof}\noindent
In addition to the information about paths, we will also need a way to find a 0-valued edge with a specific marking.
\begin{lemma}\label{critpath:edgedict}
    Within $O(m)$ space, and without further asymptotic cost in connection with an update, we can maintain information that allows us to find a $0$-valued $(u,v)$-marked edge if it exists. Our retrieval time is $O(\log(m))$.
\end{lemma}
\begin{proof}
    We maintain a search tree of size $\BigO(m)$ and height $\BigO(\log m)$, which stores the markings of $\Tree(G)$ and a 0-valued edge if it exists. If there is a dashed 0-valued edge this takes precedence. The tree may contain out-dated information, but in that case, no $0$-valued edge exists, and this can be verified by inspection in constant time. 

    Upon insertion of a new edge into $\Tree(G)$, i.e. when using \V{cut-exposed}, look up the marking in the search tree and insert the edge if the marking is not present, or replace the representative if it has positive value. 
    To maintain dashed precedence, we will in each P-node store a pointer to the relevant marking in the search tree (this pointer is easily created by a lookup during insertion). Then during $\V{splice}(\mu)$, check if $\mu$ or $p(\mu)$ is a P-node. If one of them is, check if $(\mu,p(\mu))$ is the representative of the marking. If this is the case replace it with an arbitrary 0-valued neighbour of the P-node (such a neighbour is easily maintained in the root of the child set) and if no such neighbour exists there is no other 0-valued edge with this marking, making the original representative the only candidate. In \V{slice}$(\mu)$ a dashed heavy edge might be the representative, in which case we find a suitable replacement or prove that no such replacement exists as before.

    What remains is to find a $0$-valued edge that was created by \V{split}ing an S or P-node. When splitting a P-node, its marking is already represented, and we can follow the pointer and update to representative in the same location in constant time. When splitting an $S$-node, we store along with the split S-node, its at most two secondary children. These $0$-valued edges will not be stored, but can instead be found in $O(\log m)$ time, analogously to \cref{spqr:pathinfo} by exposing the representatives of its markings and using the resulting solid path to find one of the endpoints. 

    We inspect the poles of the markings of the two vertices. We argue that one of these poles must be an endpoint of the edge. Once we have found an endpoint of thee edge, we can find the edge as edge either the parental edge of that vertex, or by inspecting its at most two secondary children. 
    There are three cases, depending on whether the representatives are on the root-side of the split (``up'') or not (``down''). 
    Let $S_1$ denote the secondary child of of the node $S_0$ that was split. 
    If both representatives are down, both poles are $S_1$. If they are one different sides of the edge,  then $S_1$ will be the pole closest to the one that is down. If they are both up, then $S_0$ will be a pole.
\end{proof}
Note that a contract might make multiple 0-valued edges positive, leaving them as tombstones in the search tree. Furthermore, as the positive values are propagated lazily, one must verify that a 0-valued edge is indeed still 0-valued by traversing the root path in $\Tree(G)$ in $\BigO(\log n)$ time. 
With this we have enough information at hand to find the critical path of some vertex pair.
\begin{lemma}\label{critpath:find}
        Let $u,v\in V(G)$ and $u\in \Gamma(\mu)$ and $v\in \Gamma(\nu)$. We can then find the endpoints of $m(u,v)$ in $\Tree(G)$ using a constant number of heavy path operations.
\end{lemma}
\begin{proof}
    First apply \cref{critpath:edgedict} to find a (dashed) 0-valued edge if it exists. If it does, we are in case 2. or 3. from \cref{spqr:critpath}. These can be distinguished by looking at the endpoints of the edge. This leaves cases 1. and 4. To distinguish these make $\mu$ the root of $\Tree(G)$ and expose $\nu$ and consider the paths $P=\mu\dots\nu$, $U_u(P)$ and $L_v(P)$. Let $\mu^\disk$ be the lower endpoint of $U_u(P)$ and $\nu^\disk$ be the upper endpoint of $L_v(P)$. We then have three main cases which can be distinguished by comparing the lengths of $U_u(P)$ and $L_v(P)$ against $|P|$:
    \begin{enumerate}[label=(\Alph*)]
        \item $\mu^\disk$ is above $\nu^\disk$.
        \item $\mu^\disk=\nu^\disk$.
        \item $\mu^\disk$ is below $\nu^\disk$.
    \end{enumerate}
    We can now determine the endpoints of the heavy path using \cref{spqr:pathinfo}.
    \begin{enumerate}[label=(\Alph*)]
        \item If $\mu^\disk$ is above $\nu^\disk$ then $U_u(P)\cap L_v(P)$ is empty. This leaves two possibilities, either there is a 0-valued edge on $\mu^\disk\dots \nu^\disk$, in which case $m(u,v)$ is $\mu^\disk \dots \nu^\disk$ or there is none, meaning that $u$ and $v$ share an R-node. If this is the case, this R-node must be $m(u,v)$ as there are no 0-valued $(u,v)$-marked edges.
        \item If $\mu^\disk=\nu^\disk$ and is not obsolete, then $\mu^\disk=m(u,v)$. If it is obsolete, the non-existence of any 0-valued $(u,v)$-marked edge ensures that the R-node it represents is $m(u,v)$. 
        \item If $\mu^\disk$ is below $\nu^\disk$ there is some subpath of $\mu\dots \nu$ only marked by $u$ and $v$. This subpath cannot contain any 0-valued edges and thus $u$ and $v$ must share a single R-node which is also the critical path. 
    \end{enumerate}
\end{proof}

%% file: tex/HPQueryAndOperation.tex
\subsection{Implementing SPQR-tree operations}
We now show how to implement the operations on our representation of the relaxed SPQR-tree of some graph. 

\begin{lemma}
    Given vertices $u$ and $v$ with representatives $\mu$ and $\nu$ and a pre-split SPQR-tree rooted at $\mu$, we can find $m(u,v)$ and provide a heavy path representing it in a constant number of heavy path operations.
\end{lemma}
\begin{proof}
First expose $\nu$ and then apply \cref{critpath:find} to find the endpoints of $m(u,v)$. Conceal $\nu$, evert the lower endpoint and expose the upper endpoint. 
\end{proof}
\begin{lemma}
    Given vertices $u$ and $v$ with representatives $\mu$ and $\nu$ and a pre-split SPQR-tree rooted at $\nu$, we can support \V{Triconnected}$(u,v)$ queries in a constant number of heavy path operations.
\end{lemma}
\begin{proof}
    We find the critical path $m(u,v)$ and consider the endpoints. We answer positively if and only if there is a 0-valued $(u,v)$-marked edge or $u$ and $v$ share a P or R-node. Correctness is a direct consequence of \cref{spqr:triconnected}.
\end{proof}
 
\begin{lemma}
    Given a pre-split SPQR-tree, we can support \V{Subdivide}$(u,v)$ in a constant number of heavy path operations.
\end{lemma}
\begin{proof}
    First, we find $m(u,v)$. As the edge $(u,v)$ exists in $G$, there must be some node that contains both $u$ and $v$. We go case by case with respect to \cref{spqr:critpath}. Let $\chi$ be a new S-node with the cycle $u \ x \ v$ as its skeleton.
    \begin{enumerate}
        \item[1.] If $m(u,v)=\mu$ is exactly one SPQR-node, then $(u,v)$ is non-virtual in $\mu$. If $\mu$ is an S-node subdivide this edge otherwise use the necessary expose and conceals to connect $\mu$ and $\chi$ with link-exposed.
        \item[3.] If $m(u,v)$ is a single P-node, connect it to $\chi$ using link-exposed and the necessary expose and conceal operations.\qedhere
    \end{enumerate}
    We let the representative of $x$ be $\chi$. The running time is evident as we only use a constant amount of heavy path operations. 
\end{proof}\noindent
We can now describe the insertion of an edge.
\begin{lemma}
    Given vertices $u$ and $v$ with representatives $\mu$ and $\nu$ and a pre-split SPQR-tree rooted at $\nu$, we can support \V{InsertEdge}$(u,v)$ in a constant number of heavy path operations.
\end{lemma}
\begin{proof}
    This is essentially a direct implementation of the procedure from \cref{spqr}. First find $m(u,v)$, then 
    \begin{enumerate}
        \item If $m(u,v)$ is an S-node, select $(u,v)$ on $m(u,v)$ and relabel the node with a P. The secondary children of $m(u,v)$ are now primary without an R-node parent and must thus not be the lazy tail. If they are, insert them properly in their heavy path and split them to restore \cref{inv3}. 
        If $m(u,v)$ is an R-node, expose the $(u,v)$-poles and associate the new edge to one of the poles. 
        \item If $m(u,v)$ is a P-node insert $(u,v)$ in the skeleton.
        \item If $m(u,v)$ is $(u,v)$-marked, split the solid path, create a new P-node $\omega$ with $(u,v)$ as its skeleton, and insert the P-node in between.  
        \item If $m(u,v)$ is not $(u,v)$-marked, make $m(u,v)$ exposed and make $u$ and $v$ selected, then use contract on the endpoints.
    \end{enumerate}
    After the insertion use conceal on the lowest endpoint of the heavy path. 
    The first three cases essentially implement the original insertion procedure. That the last case is correct warrants an argument. We go through the procedure and argue that each step is completed.
    \begin{enumerate}
        \item All the edges of $m(u,v)$ are removed as part of contract. The P-nodes have been split in such a way that their secondary children contain all the edges that aren't on the path.
        \item All of the R-nodes have no child groups and can thus be merged into the new R-node.
        \item All internal S-nodes have child groups corresponding to the split between their neighbours, and thus their secondary children correspond (almost) exactly to the result of a split. In the endpoints of the root path, the child groups also correspond to splitting by the path neighbour and $u$ or $v$. The only mismatch between the proper child groups and the result of a split is the possibility of children with $k<3$ edges, as these are by definition not proper. These are treated in case 5.
        \item  The S-nodes and P-nodes in question exist as dashed secondary children of the path and are thus, after contraction, already connected to $R^*$.
        \item This step would correspond to merging all secondary children with fewer than 3 edges. These correspond to child groups with at most one edge, which are categorically excluded from the definition and are thus already part of the path.
    \end{enumerate}
    The running time is evident as we use a constant amount of heavy path operations.
\end{proof}\noindent
By letting all vertices have unit-weight, we obtain
\begin{restate}{\cref{result:biconn}}
\theoremonetext
\end{restate}

In addition to this, we have the following extensions: 
\begin{lemma}
    Given a relaxed SPQR-tree, where some edge $(u,v)$ has just been inserted, we can remove said edge in $O(\log n)$ time.
\end{lemma}
\begin{proof}
As the insertion of an edge only creates $\BigO(\log n)$ changes, we can simply store a log of the memory operations used to insert the edge, which can then be used to bring the data structure back to the previous state when required. 
\end{proof}
\begin{lemma}
    It is possible to maintain an SPQR-tree, rooted at an S-node which supports split, merge and insert in $O(\log W)$ time.
\end{lemma}
\begin{proof}
The representation of the child set of an S-node directly allows this.
\end{proof}
We can now use the construction by Holm and Rotenberg (see \cite{HolmRotenberg2019} Section 5) to maintain a block-cut tree of a general graph $G$, which supports all operations in worst-case $O(\log n)$ calls to the SPQR-tree operations above. This construction uses the same building blocks as the pre-split SPQR-tree, and represents each block of the block-cut tree as an SPQR-tree. Every heavy path then simulates an edge between the extremal blocks, which allows fast insertion of such an edge once the blocks of its endpoints have been exposed. This yields
\begin{restate}{\cref{result:general}}
    \theoremtwotext
\end{restate}

%% file: tex/40ending.tex
\section{Conclusion and Open Problems}
We have described a data structure which maintains a biconnected graph while supporting triconnectivity queries, edge insertion and subdivision in worst-case $\BigO(\log n)$ time. The data structure does this by maintaining a representation of the SPQR-tree of the underlying graph, which is at all times ready to insert an edge with only a logarithmic number of changes to the tree. Plugging our result into the construction of Holm and Rotenberg~\cite{HolmRotenberg2019} then gives an algorithm for incremental triconnectivity in general graphs in worst-case $\BigO(\log^2 n)$ time per query and operation, improving the previous upper bound from $\BigO(\log^3 n)$ to $\BigO(\log^2 n)$.

Our result for biconnected graphs comes within an $\BigO(\log \log n)$ factor of the cell-probe incremental connectivity lower bound \cite{fredmansaks89}, and thus, a natural open problem is whether there exists an algorithm which matches this lower bound for biconnected graphs. Another open problem is the existence of an algorithm for triconnectivity in general graphs analogous to the $\BigO(\log n)$ upper bound for biconnectivity~\cite{PoutreW98}. In the fully dynamic setting, there is still no known data structure with polylogarithmic update- and query time, that can answer triconnectivity queries between vertex pairs.

\section{Declaration of AI Usage}
The paper is based on a M.Sc. thesis written without the usage of artificial intelligence or language models. In connection with preparing the manuscript for submission, the authors have used Claude Opus 5.5 and Google Gemini 4 Argon (Deep Review) for proof reading for mathematical rigour and clarity of exposition. The theorems in the paper are identical to those presented in the non-AI-assisted masters thesis. 

%% file: tex/input-til-referenceliste.tex
\nocite{Frederickson85}
\nocite{}

\nocite{}

\nocite{}
%
\nocite{PoutreW98,HolmRotenberg2019,HolmHR23,HolmR24}

%% file: refs.bib
@inproceedings{HolmRotenberg2019,
  author       = {Jacob Holm and
                  Eva Rotenberg},
  editor       = {Shuchi Chawla},
  title        = {Worst-Case Polylog Incremental SPQR-trees: Embeddings, Planarity,
                  and Triconnectivity},
  booktitle    = {Proceedings of the 2020 {ACM-SIAM} Symposium on Discrete Algorithms,
                  {SODA} 2020, Salt Lake City, UT, USA, January 5-8, 2020},
  pages        = {2378--2397},
  publisher    = {{SIAM}},
  year         = {2020},
  url          = {https://doi.org/10.1137/1.9781611975994.146},
  doi          = {10.1137/1.9781611975994.146},
  bibsource    = {dblp computer science bibliography, https://dblp.org}
}

@article{BentSleatorTarjan1985,
  title     = {Biased Search Trees},
  author    = {Samuel W. Bent and Daniel D. Sleator and Robert E. Tarjan},
  journal   = {SIAM Journal on Computing},
  volume    = {14},
  number    = {3},
  pages     = {545--568},
  year      = {1985},
  month     = {8},
  publisher = {Society for Industrial and Applied Mathematics},
  doi       = {10.1137/0214037}
}

@article{SleatorTarjan1983,
  author    = {Daniel D. Sleator and Robert E. Tarjan},
  title     = {A Data Structure for Dynamic Trees},
  journal   = {Journal of Computer and System Sciences},
  volume    = {26},
  number    = {3},
  pages     = {362--391},
  year      = {1983},
  publisher = {Academic Press},
  doi       = {10.1016/0022-0000(83)90006-5}
}

@article{top1,
  author       = {Stephen Alstrup and
                  Jacob Holm and
                  Kristian de Lichtenberg and
                  Mikkel Thorup},
  title        = {Maintaining information in fully dynamic trees with top trees},
  journal      = {{ACM} Trans. Algorithms},
  volume       = {1},
  number       = {2},
  pages        = {243--264},
  year         = {2005},
  url          = {https://doi.org/10.1145/1103963.1103966},
  doi          = {10.1145/1103963.1103966},
  bibsource    = {dblp computer science bibliography, https://dblp.org}
}

@article{Hopcroft1973,
  title={Dividing a Graph into Triconnected Components},
  author={John E. Hopcroft and Robert Endre Tarjan},
  journal={SIAM J. Comput.},
  year={1973},
  volume={2},
  pages={135-158},
  url={https://api.semanticscholar.org/CorpusID:26605417}
}

@inproceedings{DecrementalSPQR,
title = "Decremental SPQR-trees for planar graphs",
author = "Jacob Holm and Italiano, {Giuseppe F.} and Adam Karczmarz and Jakub {\L}acki and Eva Rotenberg",
year = "2018",
month = aug,
day = "1",
doi = "10.4230/LIPIcs.ESA.2018.46",
isbn = "9783959770811",
series = "Leibniz International Proceedings in Informatics, LIPIcs",
booktitle = "26th European Symposium on Algorithms, ESA 2018",
}

@article{Battista1996,
author = {Di Battista, G. and Tamassia, R.},
title = {On-line maintenance of triconnected components with SPQR-trees},
year = {1996},
issue_date = {April     1996},
publisher = {Springer-Verlag},
address = {Berlin, Heidelberg},
volume = {15},
number = {4},
issn = {0178-4617},
url = {https://doi.org/10.1007/BF01961541},
doi = {10.1007/BF01961541},
journal = {Algorithmica},
month = apr,
pages = {302–318},
numpages = {17}
}

@inbook{BridgeFinding,
	author = {Jacob Holm and Eva Rotenberg and Mikkel Thorup},
	title = {Dynamic {B}ridge-{F}inding in {$\widetilde{O}(\log^2 n)$} {A}mortized {T}ime},
	booktitle = {Proceedings of the 2018 Annual ACM-SIAM Symposium on Discrete Algorithms (SODA)},
	chapter = {},
	pages = {35-52},
	year = {2018},
	publisher = {{SIAM}},
	doi = {10.1137/1.9781611975031.3},
	URL = {https://epubs.siam.org/doi/abs/10.1137/1.9781611975031.3}
}

@article{Frederickson85,
	author    = {Greg N. Frederickson},
	title     = {{Data Structures for On-Line Updating of Minimum Spanning Trees, with
	Applications}},
	journal   = {{SIAM} Journal on Computing},
	year      = {1985},
	volume    = {14},
	number    = {4},
	pages     = {781--798},
	doi       = {10.1137/0214055},
	bibsource = {dblp computer science bibliography, http://dblp.org}
}

@inproceedings{Kapron:2013,
	author = {Kapron, Bruce M. and King, Valerie and Mountjoy, Ben},
	title = {Dynamic Graph Connectivity in Polylogarithmic Worst Case Time},
	booktitle = {Proceedings of the Twenty-fourth Annual ACM-SIAM Symposium on Discrete Algorithms},
	series = {SODA '13},
	year = {2013},
	isbn = {978-1-611972-51-1},
	location = {New Orleans, Louisiana},
	pages = {1131--1142},
	numpages = {12},
	url = {http://dl.acm.org/citation.cfm?id=2627817.2627898},
	acmid = {2627898},
	publisher = {Society for Industrial and Applied Mathematics},
	address = {Philadelphia, PA, USA},
}

@incollection{Wulff-Nilsen16a,
	author    = {Christian Wulff{-}Nilsen},
	title     = {Faster Deterministic Fully-Dynamic Graph Connectivity},
	booktitle = {Encyclopedia of Algorithms},
	pages     = {738--741},
	year      = {2016},
	publisher = {Springer Berlin Heidelberg},
	address	  = {Berlin, Heidelberg},
}

@article{patrascu06,
	title={Logarithmic lower bounds in the cell-probe model},
	author={P\v{a}tra\c{s}cu, Mihai and Demaine, Erik D.},
	journal={SIAM Journal on Computing},
	volume={35},
	number={4},
	pages={932--963},
	year={2006},
	publisher={Society for Industrial and Applied Mathematics}
}

@article{Tarjan:75,
	author = {Tarjan, Robert E.},
	title = {{Efficiency of a Good But Not Linear Set Union Algorithm}},
	journal = {Journal of the ACM},
	issue_date = {April 1975},
	volume = {22},
	number = {2},
	month = apr,
	year = {1975},
	issn = {0004-5411},
	pages = {215--225},
	numpages = {11},
	OPTurl = {http://doi.acm.org/10.1145/321879.321884},
	doi = {10.1145/321879.321884},
	acmid = {321884},
	publisher = {ACM},
	address = {New York, NY, USA},
}

@MISC{Henzinger97,
	author = {Monika Rauch Henzinger and Valerie King},
	title = {{Fully Dynamic 2-Edge Connectivity Algorithm in Polylogarithmic Time per Operation}},
    institution = {},
	year = {1997}
}

@article{Eppstein97,
	author = {Eppstein, David and Galil, Zvi and Italiano, Giuseppe F. and Nissenzweig, Amnon},
	title = {Sparsification---{A} {T}echnique {F}or {S}peeding {U}p {D}ynamic {G}raph {A}lgorithms}, 
	journal = {Journal of the ACM},
	issue_date = {Sept. 1997},
	volume = {44},
	number = {5},
	month = sep,
	year = {1997},
	issn = {0004-5411},
	pages = {669--696},
	numpages = {28},
	doi = {10.1145/265910.265914},
	acmid = {265914},
	publisher = {ACM},
	address = {New York, NY, USA},
}

@article{Frederickson97,
	author = {Frederickson, Greg N.},
	title = {Ambivalent {D}ata {S}tructures for {D}ynamic 2-{E}dge-{C}onnectivity and {$k$} {S}mallest {S}panning {T}rees},
	journal = {SIAM Journal on Computing},
	volume = {26},
	number = {2},
	pages = {484-538},
	year = {1997},
	doi = {10.1137/S0097539792226825},
}

@article{Eppstein99,
	author = {Eppstein, David and Galil, Zvi and Italiano, Giuseppe F. and Spencer, Thomas H.},
	title = {Separator-{B}ased {S}parsification {II}: {E}dge and {V}ertex {C}onnectivity},
	journal = {SIAM Journal on Computing},
	issue_date = {Feb. 1999},
	volume = {28},
	number = {1},
	month = feb,
	year = {1999},
	issn = {0097-5397},
	pages = {341--381},
	numpages = {41},
	OPTurl = {https://doi.org/10.1137/S0097539794269072},
	doi = {10.1137/S0097539794269072},
	acmid = {298566},
	publisher = {Society for Industrial and Applied Mathematics},
	address = {Philadelphia, PA, USA},
}

@article{HeTh97,
	author    = {Monika Rauch Henzinger and
	Mikkel Thorup},
	title     = {Sampling to provide or to bound: With applications to fully dynamic graph algorithms},
	journal   = {Random Struct. Algorithms},
	volume    = {11},
	number    = {4},
	pages     = {369--379},
	year      = {1997},
	doi       = {10.1002/(SICI)1098-2418(199712)11:4<369::AID-RSA5>3.0.CO;2-X},
	bibsource = {dblp computer science bibliography, http://dblp.org}
}

@article{Henzinger95fully,
	author       = {Monika Rauch Henzinger},
	title        = {Fully Dynamic Biconnectivity in Graphs},
	journal      = {Algorithmica},
	volume       = {13},
	number       = {6},
	pages        = {503--538},
	year         = {1995},
	url          = {https://doi.org/10.1007/BF01189067},
	doi          = {10.1007/BF01189067},
	bibsource    = {dblp computer science bibliography, https://dblp.org}
}

@InProceedings{Henzinger95,
	author="Henzinger, Monika R.
	and La Poutr{\'e}, Han",
	editor="Spirakis, Paul",
	title="Certificates and fast algorithms for biconnectivity in fully-dynamic graphs",
	booktitle="Algorithms --- ESA '95",
	year="1995",
	publisher="Springer Berlin Heidelberg",
	address="Berlin, Heidelberg",
	pages="171--184",
	isbn="978-3-540-44913-3"
}

@inproceedings{HuangHKP17,
	author    = {Shang{-}En Huang and
	Dawei Huang and
	Tsvi Kopelowitz and
	Seth Pettie},
	title     = {Fully {D}ynamic {C}onnectivity in {$O(\log n (\log \log n)^2)$} {A}mortized {E}xpected {T}ime},
	booktitle = {Proceedings of the Twenty-Eighth Annual {ACM-SIAM} Symposium on Discrete
	Algorithms, {SODA} 2017, Barcelona, Spain, Hotel Porta Fira, January
	16-19},
	pages     = {510--520},
	year      = {2017},
	OPTcrossref  = {DBLP:conf/soda/2017},
	OPTurl       = {https://doi.org/10.1137/1.9781611974782.32},
	doi       = {10.1137/1.9781611974782.32},
	bibsource = {dblp computer science bibliography, https://dblp.org}
}

@inproceedings{Eppstein:2003,
	author = {Eppstein, David},
	title = {{Dynamic Generators of Topologically Embedded Graphs}},
	booktitle = {Proceedings of the Fourteenth Annual ACM-SIAM Symposium on Discrete Algorithms},
	series = {SODA '03},
	year = {2003},
	isbn = {0-89871-538-5},
	location = {Baltimore, Maryland},
	pages = {599--608},
	numpages = {10},
	url = {http://dl.acm.org/citation.cfm?id=644108.644208},
	acmid = {644208},
	publisher = {Society for Industrial and Applied Mathematics},
	address = {Philadelphia, PA, USA},
}

@inproceedings{HolmHR23,
	author       = {Jacob Holm and
	Ivor {van der Hoog} and
	Eva Rotenberg},
	editor       = {Erin W. Chambers and
	Joachim Gudmundsson},
	title        = {{Worst-Case Deterministic Fully-Dynamic Biconnectivity in Changeable
	Planar Embeddings}},
	booktitle    = {39th International Symposium on Computational Geometry, SoCG 2023,
	June 12-15, 2023, Dallas, Texas, {USA}},
	series       = {LIPIcs},
	volume       = {258},
	pages        = {40:1--40:18},
	publisher    = {Schloss Dagstuhl - Leibniz-Zentrum f{\"{u}}r Informatik},
	year         = {2023},
	url          = {https://doi.org/10.4230/LIPIcs.SoCG.2023.40},
	doi          = {10.4230/LIPICS.SOCG.2023.40},
	bibsource    = {dblp computer science bibliography, https://dblp.org}
}

@inproceedings{Thorup00,
	author       = {Mikkel Thorup},
	editor       = {F. Frances Yao and
	Eugene M. Luks},
	title        = {Near-optimal fully-dynamic graph connectivity},
	booktitle    = {Proceedings of the Thirty-Second Annual {ACM} Symposium on Theory
	of Computing, May 21-23, 2000, Portland, OR, {USA}},
	pages        = {343--350},
	publisher    = {{ACM}},
	year         = {2000},
	url          = {https://doi.org/10.1145/335305.335345},
	doi          = {10.1145/335305.335345},
	bibsource    = {dblp computer science bibliography, https://dblp.org}
}

@inproceedings{ChuzhoyGLNPS20,
	author       = {Julia Chuzhoy and
	Yu Gao and
	Jason Li and
	Danupon Nanongkai and
	Richard Peng and
	Thatchaphol Saranurak},
	editor       = {Sandy Irani},
	title        = {{A Deterministic Algorithm for Balanced Cut with Applications to Dynamic
	Connectivity, Flows, and Beyond}},
	booktitle    = {61st {IEEE} Annual Symposium on Foundations of Computer Science, {FOCS}
	2020, Durham, NC, USA, November 16-19, 2020},
	pages        = {1158--1167},
	publisher    = {{IEEE}},
	year         = {2020},
	url          = {https://doi.org/10.1109/FOCS46700.2020.00111},
	doi          = {10.1109/FOCS46700.2020.00111},
	bibsource    = {dblp computer science bibliography, https://dblp.org}
}

@article{Henzinger00,
	author       = {Monika Rauch Henzinger},
	title        = {{Improved Data Structures for Fully Dynamic Biconnectivity}},
	journal      = {{SIAM} J. Comput.},
	volume       = {29},
	number       = {6},
	pages        = {1761--1815},
	year         = {2000},
	url          = {https://doi.org/10.1137/S0097539794263907},
	doi          = {10.1137/S0097539794263907},
	bibsource    = {dblp computer science bibliography, https://dblp.org}
}

@inproceedings{Rauch94,
	author       = {Monika Rauch},
	editor       = {Frank Thomson Leighton and
	Michael T. Goodrich},
	title        = {Improved data structures for fully dynamic biconnectivity},
	booktitle    = {Proceedings of the Twenty-Sixth Annual {ACM} Symposium on Theory of
	Computing, 23-25 May 1994, Montr{\'{e}}al, Qu{\'{e}}bec,
	Canada},
	pages        = {686--695},
	publisher    = {{ACM}},
	year         = {1994},
	url          = {https://doi.org/10.1145/195058.195434},
	doi          = {10.1145/195058.195434},
	bibsource    = {dblp computer science bibliography, https://dblp.org}
}

@article{menger,
    author = {Menger, Karl},
    journal = {Fundamenta Mathematicae},
    number = {1},
    pages = {96--115},
    title = {{Zur allgemeinen Kurventheorie}},
    url = {https://doi.org/10.4064/FM-10-1-96-115},
    doi = {10.4064/FM-10-1-96-115},
    volume = {10},
    year = {1927},
}

@article{PoutreW98,
  author       = {Johannes A. La Poutr{\'{e}} and
                  Jeffery R. Westbrook},
  title        = {{Dynamic 2-Connectivity with Backtracking}},
  journal      = {{SIAM} J. Comput.},
  volume       = {28},
  number       = {1},
  pages        = {10--26},
  year         = {1998},
  url          = {https://doi.org/10.1137/S0097539794272582},
  doi          = {10.1137/S0097539794272582},
  bibsource    = {dblp computer science bibliography, https://dblp.org}
}

@article{WestbrookT92,
  author       = {Jeffery R. Westbrook and
                  Robert Endre Tarjan},
  title        = {{Maintaining Bridge-Connected and Biconnected Components On-Line}},
  journal      = {Algorithmica},
  volume       = {7},
  number       = {5{\&}6},
  pages        = {433--464},
  year         = {1992},
  url          = {https://doi.org/10.1007/BF01758773},
  doi          = {10.1007/BF01758773},
  bibsource    = {dblp computer science bibliography, https://dblp.org}
}

@article{HolmR24,
	author       = {Jacob Holm and
	Eva Rotenberg},
	title        = {Good r-divisions Imply Optimal Amortized Decremental Biconnectivity},
	journal      = {Theory Comput. Syst.},
	volume       = {68},
	number       = {4},
	pages        = {1014--1048},
	year         = {2024},
	url          = {https://doi.org/10.1007/s00224-024-10181-z},
	doi          = {10.1007/S00224-024-10181-Z},
	bibsource    = {dblp computer science bibliography, https://dblp.org}
}

@article{tarjan79,
  author       = {Robert Endre Tarjan},
  title        = {A Class of Algorithms which Require Nonlinear Time to Maintain Disjoint
                  Sets},
  journal      = {J. Comput. Syst. Sci.},
  volume       = {18},
  number       = {2},
  pages        = {110--127},
  year         = {1979},
  url          = {https://doi.org/10.1016/0022-0000(79)90042-4},
  doi          = {10.1016/0022-0000(79)90042-4},
  bibsource    = {dblp computer science bibliography, https://dblp.org}
}

@inproceedings{fredmansaks89,
  author       = {Michael L. Fredman and
                  Michael E. Saks},
  editor       = {David S. Johnson},
  title        = {The Cell Probe Complexity of Dynamic Data Structures},
  booktitle    = {Proceedings of the 21st Annual {ACM} Symposium on Theory of Computing,
                  May 14-17, 1989, Seattle, Washington, {USA}},
  pages        = {345--354},
  publisher    = {{ACM}},
  year         = {1989},
  url          = {https://doi.org/10.1145/73007.73040},
  doi          = {10.1145/73007.73040},
  bibsource    = {dblp computer science bibliography, https://dblp.org}
}

@article{blum86,
  author       = {Norbert Blum},
  title        = {On the Single-Operation Worst-Case Time Complexity of the Disjoint
                  Set Union Problem},
  journal      = {{SIAM} J. Comput.},
  volume       = {15},
  number       = {4},
  pages        = {1021--1024},
  year         = {1986},
  url          = {https://doi.org/10.1137/0215072},
  doi          = {10.1137/0215072},
  bibsource    = {dblp computer science bibliography, https://dblp.org}
}

@inproceedings{alstrupBR99,
  author       = {Stephen Alstrup and
                  Amir M. Ben{-}Amram and
                  Theis Rauhe},
  editor       = {Jeffrey S. Vitter and
                  Lawrence L. Larmore and
                  Frank T. Leighton},
  title        = {Worst-Case and Amortised Optimality in Union-Find (Extended Abstract)},
  booktitle    = {Proceedings of the Thirty-First Annual {ACM} Symposium on Theory of
                  Computing, May 1-4, 1999, Atlanta, Georgia, {USA}},
  pages        = {499--506},
  publisher    = {{ACM}},
  year         = {1999},
  url          = {https://doi.org/10.1145/301250.301383},
  doi          = {10.1145/301250.301383},
  bibsource    = {dblp computer science bibliography, https://dblp.org}
}

@book{tutte66,
 ISBN = {9781487572969},
 URL = {http://www.jstor.org/stable/10.3138/j.ctvfrxhc8},
 author = {William T. Tutte},
 publisher = {University of Toronto Press},
 title = {Connectivity in Graphs},
 urldate = {2026-09-23},
 year = {1966}
}

@article{HenzingerK99,
  author       = {Monika R. Henzinger and
                  Valerie King},
  title        = {Randomized Fully Dynamic Graph Algorithms with Polylogarithmic Time
                  per Operation},
  journal      = {J. {ACM}},
  volume       = {46},
  number       = {4},
  pages        = {502--516},
  year         = {1999},
  url          = {https://doi.org/10.1145/320211.320215},
  doi          = {10.1145/320211.320215},
  bibsource    = {dblp computer science bibliography, https://dblp.org}
}

@misc{Adena26,
      title={Incremental Exact Vertex Connectivity}, 
      author={Felix Adena and Monika Henzinger and Christoph Kern and Daniel Rutschmann},
      year={2026},
      eprint={2610.21531},
      archivePrefix={arXiv},
      primaryClass={cs.DS},
      url={https://arxiv.org/abs/2610.21531}, 
}

@article{HolmLT01,
  author       = {Jacob Holm and
                  Kristian de Lichtenberg and
                  Mikkel Thorup},
  title        = {Poly-logarithmic deterministic fully-dynamic algorithms for connectivity,
                  minimum spanning tree, 2-edge, and biconnectivity},
  journal      = {J. {ACM}},
  volume       = {48},
  number       = {4},
  pages        = {723--760},
  year         = {2001},
  url          = {https://doi.org/10.1145/502090.502095},
  doi          = {10.1145/502090.502095},
  bibsource    = {dblp computer science bibliography, https://dblp.org}
}

@article{SleatorT85,
  author       = {Daniel D. Sleator and
                  Robert E. Tarjan},
  title        = {Self-Adjusting Binary Search Trees},
  journal      = {J. {ACM}},
  volume       = {32},
  number       = {3},
  pages        = {652--686},
  year         = {1985},
  url          = {https://doi.org/10.1145/3828.3835},
  doi          = {10.1145/3828.3835},
  bibsource    = {dblp computer science bibliography, https://dblp.org}
}

@article{Whitney1932,
  author  = {Whitney, Hassler},
  title   = {Non-separable and planar graphs},
  journal = {Transactions of the American Mathematical Society},
  volume  = {34},
  number  = {2},
  pages   = {339--362},
  year    = {1932},
  doi     = {10.1090/S0002-9947-1932-1501641-2}
}

@article{maclane1937,
author = "Mac Lane, Saunders",
doi = "10.1215/S0012-7094-37-00336-3",
fjournal = "Duke Mathematical Journal",
journal = "Duke Math. J.",
month = "09",
number = "3",
pages = "460--472",
publisher = "Duke University Press",
title = "A structural characterization of planar combinatorial graphs",
url = "https://doi.org/10.1215/S0012-7094-37-00336-3",
volume = "3",
year = "1937"
}

@inproceedings{Rauch92,
  author       = {Monika Rauch},
  title        = {Fully Dynamic Biconnectivity in Graphs},
  booktitle    = {33rd Annual Symposium on Foundations of Computer Science, Pittsburgh,
                  Pennsylvania, USA, October 24-27, 1992},
  pages        = {50--59},
  publisher    = {{IEEE} Computer Society},
  year         = {1992},
  url          = {https://doi.org/10.1109/SFCS.1992.267819},
  doi          = {10.1109/SFCS.1992.267819},
  bibsource    = {dblp computer science bibliography, https://dblp.org}
}

@article{Planken25,
  author       = {Romain Bourneuf and
                  Tim Planken},
  title        = {A Structural Linear-Time Algorithm for Computing the Tutte Decomposition},
  journal      = {CoRR},
  volume       = {abs/2508.06212},
  year         = {2025},
  url          = {https://doi.org/10.48550/arXiv.2508.06212},
  doi          = {10.48550/ARXIV.2508.06212},
  eprinttype   = {arXiv},
  eprint       = {2508.06212},
  bibsource    = {dblp computer science bibliography, https://dblp.org}
}

@inproceedings{GutwengerM00,
  author       = {Carsten Gutwenger and
                  Petra Mutzel},
  editor       = {Joe Marks},
  title        = {A Linear Time Implementation of SPQR-Trees},
  booktitle    = {Graph Drawing, 8th International Symposium, {GD} 2000, Colonial Williamsburg,
                  VA, USA, September 20-23, 2000, Proceedings},
  series       = {Lecture Notes in Computer Science},
  volume       = {1984},
  pages        = {77--90},
  publisher    = {Springer},
  year         = {2000},
  url          = {https://doi.org/10.1007/3-540-44541-2},
  doi          = {10.1007/3-540-44541-2_8},
  bibsource    = {dblp computer science bibliography, https://dblp.org}
}
